\documentclass[a4paper, USenglish, numberwithinsect, cleveref, autoref, thm-restate]{lipics-v2019}
\usepackage{xparse,fdsymbol,stmaryrd}
\hypersetup{hypertexnames=false} 

\newcommand{\otyp}{\mathsf{o}}
\newcommand{\arr}{\mathbin{\to}}
\newcommand{\Arr}{\mathbin{\Rightarrow}}
\newcommand{\set}[1]{\{#1\}}
\newcommand{\dom}{\mathrm{dom}}
\newcommand{\setof}[2]{\{#1\mid#2\}}
\newcommand{\Nat}{\mathbb{N}}
\newcommand{\Aa}{\mathcal{A}}
\newcommand{\Gg}{\mathcal{G}}
\newcommand{\Hh}{\mathcal{H}}

\newcommand{\Oo}{\mathcal{O}}

\newcommand{\Rr}{\mathcal{R}}

\newcommand{\Xx}{\mathcal{X}}
\newcommand{\Yy}{\mathcal{Y}}
\newcommand{\Zz}{\mathcal{Z}}
\newcommand{\varS}{\mathsf{S}}
\newcommand{\varT}{\mathsf{T}}
\newcommand{\varX}{\mathsf{X}}
\newcommand{\varY}{\mathsf{Y}}
\newcommand{\varZ}{\mathsf{Z}}
\newcommand{\vary}{\mathsf{y}}
\newcommand{\varz}{\mathsf{z}}
\newcommand{\ord}{\mathsf{ord}}
\newcommand{\gar}{\mathsf{gar}}
\newcommand{\tp}{\mathsf{tp}}
\newcommand{\pr}{\mathsf{pr}}
\newcommand{\tr}{\mathsf{tr}}
\newcommand{\trt}{\mathsf{tr}^\mathsf{t}}
\newcommand{\RT}{\mathsf{RT}}
\newcommand{\BT}{\mathsf{BT}}
\newcommand{\BText}{\mathsf{BT}^\mathsf{ext}}
\newcommand{\BTsimpl}{\mathsf{BT}^\mathsf{s}}
\newcommand{\Eve}{\mathrm{Eve}}
\newcommand{\Adam}{\mathrm{Adam}}
\newcommand{\player}{\wp}
\newcommand{\subtree}{\!{\restriction}\!}
\newcommand{\restr}{\!{\restriction}\!}
\newcommand{\rew}{\longrightarrow}
\newcommand{\erew}{\rightsquigarrow}
\newcommand{\simpl}{\rightarrowtail}
\newcommand{\esubst}[3]{#1\llparenthesis#2/#3\rrparenthesis}
\newcommand{\esubstdots}[5]{#1\llparenthesis#2/#3\rrparenthesis\dots\llparenthesis#4/#5\rrparenthesis}
\newcommand{\symb}[1]{\langle#1\rangle}
\newcommand{\subst}[1]{[#1]}
\newcommand{\mapch}[1]{[#1]}
\newcommand{\scope}[1]{[#1]}
\renewcommand{\exp}{\mathit{exp}}

\Crefname{lemma}{Lemma}{Lemmata}
\Crefname{equation}{Equality}{Equalities}
\Crefname{equalities}{Equalities}{Equalities}\creflabelformat{equalities}{(#2#1#3)}
\creflabelformat{enumi}{(#2#1#3)}
\Crefname{case-br}{Case}{Cases}\creflabelformat{case-br}{(#2#1#3)}
\Crefname{case}{Case}{Cases}

\title{Higher-Order Model Checking Step by Step}

\author{Paweł Parys}{Institute of Informatics, University of Warsaw, Poland}{parys@mimuw.edu.pl}{https://orcid.org/0000-0001-7247-1408}{}

\authorrunning{P. Parys}

\Copyright{Paweł Parys}



\funding{Work supported by the National Science Centre, Poland (grant no.\@ 2016/22/E/ST6/00041).}

\acknowledgements{The author would like to thank the anonymous reviewers for their constructive comments.}

\relatedversion{This is an extended version of a paper published on the ICALP 2021 conference.}

\nolinenumbers
\hideLIPIcs

\begin{document}

\maketitle

\begin{abstract}
	We show a new simple algorithm that solves the model-checking problem for recursion schemes:
	check whether the tree generated by a given higher-order recursion scheme is accepted by a given alternating parity automaton.
	The algorithm amounts to a procedure that transforms a recursion scheme of order $n$ to a recursion scheme of order $n-1$, preserving acceptance, and increasing the size only exponentially.
	After repeating the procedure $n$ times, we obtain a recursion scheme of order $0$, for which the problem boils down to solving a finite parity game.
	Since the size grows exponentially at each step, the overall complexity is $n$\textsf{-EXPTIME}, which is known to be optimal.
	More precisely, the transformation is linear in the size of the recursion scheme,
	assuming that the arity of employed nonterminals and the size of the automaton are bounded by a constant;
	this results in an \textsf{FPT} algorithm for the model-checking problem.

	Our transformation is a generalization of a previous transformation of the author (2020), working for reachability automata in place of parity automata.
	The step-by-step approach can be opposed to previous algorithms solving the considered problem ``in one step'', being compulsorily more complicated.
\end{abstract}

\section{Introduction}

	Recursion schemes are faithful and algorithmically manageable abstractions of the control flow of programs involving higher-order functions~\cite{Kobayashi-jacm}.
	Such functions are nowadays widely used not only in functional programming languages such as Haskell and the OCAML family,
	but also in mainstream languages such as Java, JavaScript, Python, and C++.
	Simultaneously, the formalism of recursion schemes is equivalent via direct translations to simply-typed $\lambda Y$-calculus~\cite{lambdaY}.
	Collapsible pushdown systems~\cite{collapsible-translation} and ordered tree-pushdown systems~\cite{tree-pushdown} are other equivalent formalisms.
	Recursion schemes cover some other models such as indexed grammars~\cite{indexed} and ordered multi-pushdown automata~\cite{multi-pushdown}.

	The most celebrated algorithmic result in the analysis of recursion schemes is the decidability of the \emph{model-checking problem} against regular properties of trees:
	given a recursion scheme $\Gg$ and a parity tree automaton $\Aa$, one can decide whether the tree generated by $\Gg$ is accepted by $\Aa$~\cite{Ong-schemes}.
	This fundamental result has been reproved several times, that is, using collapsible higher-order pushdown automata~\cite{collapsible-orig},
	intersection types~\cite{KobayashiOng-types}, Krivine machines~\cite{Krivine},
	and it has been extended in diverse directions such as global model checking~\cite{global}, logical reflection~\cite{reflection}, effective selection~\cite{effective-selection},
	and a transfer theorem via models of lambda-calculus~\cite{models}.
	The model-checking problem for recursion schemes of order $n$ is complete for $n$-fold exponential time~\cite{Ong-schemes}.
	Despite this hardness result, the model-checking problem can be solved efficiently on multiple nontrivial examples,
	thanks to the development of several recursion-scheme model checkers~\cite{practical-apta, travmc2, horsats}
	(including some model checkers that work only for automata models weaker than parity tree automata~\cite{trecs, gtrecs, horsat, travmc, preface}).

	In this paper, we give a new simple algorithm solving the model-checking problem for recursion schemes, mentioned above.
	The algorithm amounts to a procedure that transforms a recursion scheme of order $n$ to a recursion scheme of order $n-1$, preserving acceptance, and increasing the size only exponentially.
	After repeating the procedure $n$ times, we obtain a recursion scheme of order $0$, for which acceptance boils down to winning a finite parity game.
	Since the size grows exponentially at each step, we reach the optimal overall complexity of $n$-fold exponential time.
	In a more detailed view, the complexity looks even better:
	the size growth is exponential only in the arity of types appearing in the recursion scheme, and in the size of the parity automaton;
	if these two parameters are bounded by a constant, the transformation is linear in the size of the recursion scheme.
	Since solving a finite parity game is \textsf{FPT} in the number of priorities~\cite{Calude},
	our algorithm for the the model-checking algorithm is \textsf{FTP} in the two parameters.%
	\footnote{This is not new.
		Actually, most previous algorithms reduce the model-checking problem to the problem of solving a parity game
		whose size is polynomial (for a polynomial of a fixed degree, for some algorithms just linear) in the size of the input,
		assuming that the arity of types appearing in the recursion scheme and the size of the parity automaton are fixed.
		Thus, only the method introduced by us is new, not the complexity results.}

	The main difference between our algorithm and all the others is that we solve the problem step by step, repeatedly reducing the order by one,
	while most previous algorithms work ``in one step'', being compulsorily more complicated.
	The only algorithms that have been reducing the order by one, were algorithms using collapsible pushdown automata~\cite{collapsible-orig,reflection,effective-selection}.
	Notice, however, that these algorithms: first, are very technical; second, are contained only in unpublished appendices and in an arXiv paper~\cite{collapsible-arxiv};
	third, if we want to use them for recursion schemes,
	it is necessary to employ a (nontrivial) translation from recursion schemes to collapsible pushdown automata~\cite{collapsible-translation,lambdaY,effective-selection}.
	A reduction of order was also possible for a subclass of recursion schemes, called \emph{safe} recursion schemes~\cite{easy-trees},
	but it was not known how to extend it to all recursion schemes.

	The transformation presented in this paper generalizes of a previous transformation of the author~\cite{trans-nonempty},
	working for reachability automata in place of parity automata.
	It has also a close relationship with a transformation given by Asada and Kobayashi~\cite{word2tree}.

\section{Preliminaries}\label{sec:prelim}

	For a number $k\in\Nat$ we write $\scope{k}$ for $\set{1,\dots,k}$.
	For any relation $\rew$ we write $\rew^*$ for the reflexive transitive closure of $\rew$.

	For a function $Z$ we write $Z\mapch{z\mapsto r}$ to denote the function that maps $z$ to $r$ while all other elements of the domain of $Z$ are mapped as in $Z$.
	Likewise, we write $Z\mapch{z_i\mapsto r_i\mid i\in I}$ to denote the function that maps $z_i$ to $r_i$ for all $i\in I$, while all other elements of the domain of $Z$ are mapped as in $Z$.
	We also use this notation without the ``$Z$'' part, for a function $Z$ with empty domain.

    \subparagraph{Recursion schemes.}

	The set of \emph{(simple) types} is constructed from a unique ground type $\otyp$ using a binary operation $\arr$;
	namely $\otyp$ is a type, and if $\alpha$ and $\beta$ are types, so is $\alpha\arr\beta$.
	By convention, $\arr$ associates to the right, that is, $\alpha\arr\beta\arr\gamma$ is understood as $\alpha\arr(\beta\arr\gamma)$.
	We often abbreviate $\underbrace{\alpha\arr\dots\arr\alpha}_\ell\to\beta$ as $\alpha^\ell\arr\beta$.
	The \emph{order} of a type $\alpha$, denoted $\ord(\alpha)$, is defined by induction:
	$\ord(\alpha_1\arr\dots\arr\alpha_k\arr\otyp)=\max(\set{0}\cup\setof{\ord(\alpha_i)+1}{i\in\scope{k}})$;
	for example $\ord(\otyp)=0$, $\ord(\otyp\arr\otyp\arr\otyp)=1$, and $\ord((\otyp\arr\otyp)\arr\otyp)=2$.

	Having a set of typed nonterminals $\Xx$, a set of typed variables $\Yy$, and a set of symbols $\Sigma$,
	\emph{terms} over $(\Xx,\Yy,\Sigma)$ are defined by induction:
	\begin{itemize}
	\item	nonterminal: every nonterminal $X\in\Xx$ of type $\alpha$ is a term of type $\alpha$;
	\item	variable: every variable $y\in\Yy$ of type $\alpha$ is a term of type $\alpha$;
	\item	node constructor: if $K_1,\dots,K_k$ are terms of type $\otyp$ and $a\in\Sigma$, then $\symb{a,K_1,\dots,K_k}$ is a term of type $\otyp$;
	\item	application: if $K$ is a term of type $\alpha\arr\beta$, and $L$ is a term of type $\alpha$, then $K\,L$ is a term of type $\beta$.
	\end{itemize}
	The type of a term $K$ is denoted $\tp(K)$.
	The order of a term $K$, written $\ord(K)$, is defined as the order of its type.

	A \emph{(higher-order) recursion scheme} is a tuple $\Gg=(\Xx,X_0,\Sigma,\Rr)$, where $\Xx$ is a finite set of typed nonterminals,
	and $X_0\in\Xx$ is a \emph{starting nonterminal} of type $\otyp$,
	and $\Sigma$ is a finite set of symbols (called an \emph{alphabet}),
	and $\Rr$ is a function assigning to every nonterminal $X\in\Xx$ a \emph{rule} of the form $X\,y_1\,\dots\,y_k\to R$,
	where $\tp(X)=(\tp(y_1)\arr\dots\arr\tp(y_k)\arr\otyp)$, and $R$ is a term of type $\otyp$ over $(\Xx,\set{y_1,\dots,y_k},\Sigma)$.
	The order of a recursion scheme, $\ord(\Gg)$, is defined as the maximum of orders of its nonterminals.

	Having a recursion scheme $\Gg=(\Xx,X_0,\Sigma,\Rr)$,
	for every set of variables $\Yy$ we define a \emph{reduction relation} $\rew_\Gg$ between terms over $(\Xx,\Yy,\Sigma)$ as the least relation such that
	\begin{itemize}
	\item	$X\,K_1\,\dots\,K_k\rew_\Gg R\subst{K_1/y_1,\dots,K_k/y_k}$ if the rule for $X$ is $X\,y_1\,\dots\,y_k\to R$,
		where $R\subst{K_1/y_1,\allowbreak\dots,K_k/y_k}$ denotes the term obtained from $R$ by substituting $K_i$ for $y_i$ for all $i\in\scope{k}$.
	\end{itemize}

	A (potentially infinite) \emph{tree} over an alphabet $\Sigma$ is defined by coinduction: every tree over $\Sigma$ is of the form $\symb{a,T_1,\dots,T_k}$,
	where $a\in\Sigma$ and $T_1,\dots,T_k$ are again trees over $\Sigma$ (for an introduction to coinductive definitions and proofs see, e.g., Czajka~\cite{Czajka}).
	We employ the usual notions of nodes, children, branches, etc.
	Formally, we can define nodes as sequences of natural numbers;
	then for a tree $T=\symb{a,T_1,\dots,T_k}$, the empty sequence $()$ is a node of $T$ labeled by $a$,
	and any longer sequence $(i_1,i_2,\dots,i_n)$ is a node of $T$ labeled by $b$ if $i_1\in\scope{k}$ and $(i_2,\dots,i_n)$ is a node of $T_{i_1}$ labeled by $b$.
	For a tree $T$ and its node $v$, we write $T\subtree_v$ for the subtree of $T$ starting at $v$.

	Again by coinduction, we define the tree \emph{generated} by a recursion scheme $\Gg=(\Xx,X_0,\Sigma,\allowbreak\Rr)$
	from a term $M$ of type $\otyp$ (over $(\Xx,\emptyset,\Sigma)$), denoted $\BT_\Gg(M)$:
	\begin{itemize}
	\item	if $M\rew_\Gg^*\symb{a,K_1,\dots,K_k}$, then $\BT_\Gg(M)=\symb{a,\BT_\Gg(K_1),\dots,\BT_\Gg(K_k)}$;
	\item	otherwise, $\BT_\Gg(M)=\symb{\omega}$ for a special symbol $\omega\not\in\Sigma$.
	\end{itemize}
	The tree generated by $\Gg$ (without mentioning a term), denoted $\BT(\Gg)$, is defined as $\BT_\Gg(X_0)$.

    \subparagraph{Parity games.}

	As already said, in the model-checking problem we are given a recursion scheme $\Gg$ and an alternating parity automaton $\Aa$,
	and we are asked whether the tree $T_\Gg$ generated by $\Gg$ is accepted by $\Aa$.
	One can, however, create a product of $\Gg$ and $\Aa$, which is a recursion scheme $\Gg_\Aa$ generating the tree of all possible runs of $\Aa$ on $T_\Gg$.
	This tree is a parity game;
	the game is won by Eve if and only if $\Aa$ accepts $T_\Gg$ (see \cref{app:product} for more details).
	Due to this reduction, it is enough to work with recursion schemes generating parity games, and consider the problem of finding a winner in such games.

	For every $d\in\Nat_+$ we consider the alphabet $\Sigma_d=\set{\Adam,\Eve}\times\scope{d}$.
	A \emph{parity tree} is a tree over $\Sigma_d$ where every node has at least one child.
	A \emph{parity recursion scheme} is a recursion scheme generating a parity tree (in particular the generated tree cannot have nodes without children, including $\omega$-labeled nodes).
	For a node labeled by $(\player,p)\in\Sigma_d$, we say that it \emph{belongs} to the player $\player$, and that it has \emph{priority} $p$.
	For trees and terms we write $\symb{\player,p,K_1,\dots,K_k}$ instead of $\symb{(\player,p),K_1,\dots,K_k}$, avoiding excessive brackets.

	A branch $\xi$ in a parity tree $T$ is \emph{won by Eve} (\emph{Adam}) if the greatest priority appearing infinitely often on $\xi$ is even (odd, respectively).
	A \emph{strategy} $\rho$ of a player $\player\in\set{\Adam,\Eve}$ in a parity tree $T$ is a function that assigns numbers to nodes of $T$ belonging to the player $\player$;
	if a node $v$ has $k$ children, we require that $\rho(v)\in\scope{k}$.
	A branch $\xi$ \emph{agrees} with $\rho$ if for every node $v$ on $\xi$ that belongs to $\player$, the next node of $\xi$ is the $\rho(v)$-th child of $v$.
	A strategy $\rho$ of $\player$ is \emph{winning} if all branches that agree with $\rho$ are winning for $\player$.
	Finally, $\player$ \emph{wins} in $T$ if $\player$ has a winning strategy in $T$; otherwise $\player$ \emph{loses} in $T$.
	It is a standard result that in every parity tree $T$ exactly one of the players wins.
	
	It is useful to consider the following order $\preceq$ on positive natural numbers (priorities): $\dots\preceq 5\preceq 3\preceq 1\preceq 2\preceq 4\preceq 6\preceq\dots$
	(first we have odd numbers in the reversed order, and then positive even numbers).
	We use the words \emph{worse} and \emph{better} to say that a priority is, respectively, earlier or later in this order.
	The intuition is that while playing a parity game, Eve always prefers to see better priorities.

\section{Transformation}\label{sec:transformation}

	In this section we present a transformation, called \emph{order-reducing transformation}, resulting in the main theorem of this paper:

	\begin{theorem}\label{thm:main}
		For any $n\geq 1$, there exists a transformation from order-$n$ parity recursion schemes to order-$(n-1)$ parity recursion schemes,
		and a polynomial $p_n$ such that, for any order-$n$ parity recursion scheme $\Gg$,
		the winner in the tree generated by the resulting recursion scheme $\Gg^\dag$ is the same as in the tree generated by $\Gg$, and $|\Gg^\dag|\leq 2^{p_n(|\Gg|)}$.
	\end{theorem}

    \subparagraph{Intuitions.}

	Let us first present intuitions behind our transformation.
	While reducing the order, we have to replace, in particular, order-$1$ functions by order-$0$ terms.
	Consider for example a tree $T$ generated from a term $K\,L$ of type $\otyp$, where $K$ has type $\otyp\arr\otyp$.
	Essentially, $T$ consists of a context $C_K$, generated by $K$, where the tree $T_L$ generated by $L$ is inserted in some ``holes''.
	Instead of playing in $T$, we propose the following modification of the game.
	At the beginning, we ask Eve a question: how is she going to reach subtrees $T_L$ while playing in $T$?
	She may declare that, according to her winning strategy,
	\begin{itemize}
	\item	she is able to ensure that the greatest priority seen before reaching $T_L$ will not be worse than $r$, for some number $r$ of her choice, or
	\item	she will not reach subtrees $T_L$ at all, which amounts to choosing for $r$ an even number greater than $d$, say $r=2d$.
	\end{itemize}
	Then, we ask Adam if he believes in this declaration.
	If so, we simply read the declared worst-case priority $r$, and we continue playing in $T_L$ (this possibility is unavailable for Adam, if Eve declared that she will not visit $T_L$).
	Otherwise, we check the declaration: we start playing in $C_K$;
	while reaching a place where $T_L$ should be placed, Eve immediately wins (loses) if her declaration is fulfilled (not fulfilled, respectively).

	We can see that such a modification of the game (even applied in infinitely many places of the considered tree) does not change the winner.
	A subtle point is that, in the modified game, Eve has to make a declaration on the priority $r$ before actually starting the game in the tree generated from $K\,L$,
	and it is not completely obvious why the need for the declaration introduces no disadvantage for Eve.
	Nevertheless, for a fixed Eve's winning strategy, the worst greatest priority seen before reaching $T_L$ is fixed, so that Eve can declare it as $r$.

	In the transformation, we change the order-$1$ term $K$ into several order-$0$ terms: $K_r$ for $r\in\set{1,\dots,d,2d}$
	(where $d$ is a bound on priorities in the considered parity recursion scheme $\Gg$).
	These terms generate trees of the same shape as the context $C_K$ generated by $K$ but with some fixed trees substituted in place of the holes of $C_K$
	(where originally trees generated by the argument $L$ were attached).
	The generated trees correspond to particular declarations made by Eve, as described above.
	Namely, we consider some fixed trees $\bot$ and $\top$ in which Eve loses and wins, respectively.
	Then, in the tree generated by $K_r$, the tree $\top$ is placed in holes such that the greatest priority on the path from the root to the hole is not worse than $r$,
	and the tree $\bot$ is placed in the remaining holes.
	In particular, the tree $\bot$ is placed in all holes of the tree generated by $K_{2d}$, because all priorities actually appearing in the tree are worse than $2d$.
	Finally, we replace $K\,L$ by $\symb{\Eve,1,K_1^L,K_2^L,\dots,K_d^L,K_{2d}}$,
	where $K_r^L=\symb{\Adam,1,K_r,\allowbreak\symb{\Eve,\allowbreak r,\allowbreak L}}$.
	In this way we realize the modified game described above: first Eve chooses a declaration $r$ and then Adam either proceed to $K_r$ or to $L$ after seeing priority $r$
	(the latter possibility is disabled for $r=2d$).
	The priority $1$ of the newly created tree nodes should be seen as a neutral priority; higher priorities visited later will be more important anyway.

	When a term $K$ of order $1$ takes multiple arguments (instead of one argument $L$), we proceed in the same way, allowing Eve to make declarations for each of the arguments.

	While applying the above-described transformation to recursion schemes,
	it is possible that the term $K$ considered above contains some nonterminals or variables.
	Then, in order to realize the transformation, we need to create multiple copies of these nonterminals and variables, corresponding to particular declarations of Eve.
	
	For example, say that in a recursion scheme we have (among others) the following two rules:
	\begin{align*}
		&\varX\to\varY\,\varZ,\\
		&\varY\,\varz\to\symb{\Eve,1,\varz,\symb{\Eve,2,\varz}}.
	\end{align*}
	Here $\varX$ and $\varZ$ are of type $\otyp$, and $\varY$ is of type $\otyp\arr\otyp$, so $\varY\,\varZ$ is an application that should be replaced by the transformation.
	Assuming $d=2$, we should obtain the following rules:
	\begin{align*}
		&\varX'\to\symb{\Eve,1,\symb{\Adam,1,\varY_1,\symb{\Eve,1,\varZ'}},
			\symb{\Adam,1,\varY_2,\symb{\Eve,2,\varZ'}},
			\varY_4},\\
		&\varY_1\to\symb{\Eve,1,\downVdash,\symb{\Eve,2,\downVdash}},\\
		&\varY_2\to\symb{\Eve,1,\upVdash,\symb{\Eve,2,\downVdash}},\\
		&\varY_4\to\symb{\Eve,1,\upVdash,\symb{\Eve,2,\upVdash}},
	\end{align*}
	where $\upVdash$ and $\downVdash$ are nonterminals from which the trees $\bot$ and $\top$ (in which Eve loses and wins, respectively) are generated.
	
	Another possibility is that in the original recursion scheme we have $\vary\,\varZ$ instead of $\varY\,\varZ$:
	\begin{align*}
		&\varS\to\varT\,\varY,\\
		&\varT\,\vary\to\vary\,\varZ.
	\end{align*}
	Then, the single parameter $\vary$ gets transformed into three parameters:
	\begin{align*}
		&\varS'\to\varT'\,\varY_1\,\varY_2\,\varY_4,\\
		&\varT'\,\vary_1\,\vary_2\,\vary_4\to\symb{\Eve,1,\symb{\Adam,1,\vary_1,\symb{\Eve,1,\varZ'}},
			\symb{\Adam,1,\vary_2,\symb{\Eve,2,\varZ'}},
			\vary_4}.
	\end{align*}

    \subparagraph{Formal definition.}

	We now formalize the above intuitions.
	Fix a parity recursion scheme $\Gg=(\Xx,X_0,\Sigma_d,\Rr)$;
	in particular fix a bound $d$ on priorities appearing in $\Gg$.
	
	A set $D_d$ of Eve's \emph{declarations} is defined as $D_d=\set{1,\dots,d,2d}$.
	For a priority $p\in\scope{d}$ and a declaration $r\in D_d$ we define a \emph{shifted} declaration $r\restr_p$ (obtained from $r$ after seeing priority $p$):
	\begin{align*}
		r\restr_p=\left\{\begin{array}{ll}
			p+1&\mbox{if $p$ is odd and $p>r$,}\\
			p-1&\mbox{if $p$ is even and $p\geq r$,}\\
			r&\mbox{otherwise.}
		\end{array}\right.
	\end{align*}
	We remark that the same definition appears in Tsukada and Ong~\cite{shift-was-here} (where shifts are called left-residuals);
	a slightly different representation is present also in Salvati and Walukiewicz~\cite{Krivine} (with declarations called residuals and shifts called liftings).
	
	The \emph{leader} (``most important priority'') of a sequence of priorities $\pi$ is the greatest priority appearing in $\pi$, or $1$ if $\pi$ is empty.
	A sequence of priorities $\pi$ \emph{fulfils} a declaration $r\in D_d$ if $r$ is worse or equal than the leader of $\pi$
	(where ``worse'' refers to the $\preceq$ order defined in \cref{sec:prelim}).
	For example, $1,4,2$, and $1,1,1$, both fulfil $3$, but $1,5,4$ does not.
	The empty sequence fulfils $r$ exactly when $r$ is odd.
	No sequence of priorities from $\scope{d}$ fulfils $2d$.
	The following \lcnamecref{shift2fulfilled} is obtained by a direct analysis (see \cref{app:shift2fulfilled}):
	
        \begin{restatable}{lemma}{shiftfulfilled}\label{shift2fulfilled}
		A sequence of priorities $p_1,p_2,\dots,p_k\in\scope{d}$ fulfils a declaration $r\in D_d$ if and only if $p_2,\dots,p_k$ fulfils $r\restr_{p_1}$.
	\end{restatable}
	
	Having a type, we are interested in cutting off its suffix of order $1$.
	Thus, we use the notation $\alpha_1\arr\dots\arr\alpha_k\Arr\otyp^\ell\arr\otyp$ for a type $\alpha_1\arr\dots\arr\alpha_k\arr\otyp^\ell\arr\otyp$ such that either $k=0$ or $\alpha_k\neq\otyp$.
	Notice that every type $\alpha$ can be uniquely represented in this form.
	We remark that some among the types $\alpha_1,\dots,\alpha_{k-1}$ (but not $\alpha_k$) may be $\otyp$.
	For a type $\alpha$ we write $\gar(\alpha)$ (``ground arity'') for the number $\ell$ for which we can write $\alpha=(\alpha_1\arr\dots\arr\alpha_k\Arr\otyp^\ell\arr\otyp)$;
	we also extend this to terms: $\gar(M)=\gar(\tp(M))$.

	We transform terms of type $\alpha$ to terms of type $\alpha^{\dag_d}$, which is defined by induction:
	\begin{align*}
		(\alpha_1\arr\dots\arr\alpha_k\Arr\otyp^\ell\arr\otyp)^{\dag_d} = \left((\alpha_1^{\dag_d})^{|D_d|^{\gar(\alpha_1)}}\arr\dots\arr(\alpha_k^{\dag_d})^{|D_d|^{\gar(\alpha_k)}}\arr\otyp\right).
	\end{align*}
	Thus, we remove all trailing order-$0$ arguments, and we multiplicate (and recursively transform) remaining arguments.
	The number of copies depends on the bound $d$ on priorities appearing in the considered parity recursion scheme.

	For a finite set $S$, we write $D_d^S$ for the set of functions $A\colon S\to D_d$.
	Moreover, we assume some fixed order on functions in $D_d^S$,
	and we write $P\,(Q_A)_{A\in D_d^S}$ for an application $P\,Q_{A_1}\,\dots\,Q_{A_{|D_d|^{|S|}}}$, where $A_1,\dots,A_{|D_d|^{|S|}}$ are all the functions from $D_d^S$ listed in the fixed order.
	The only function in $D_d^\emptyset$ is denoted $\emptyset$.

	For every variable $y$ and for every function $A\in D_d^{\scope{\gar(y)}}$ we consider a variable $y_A^{\dag_d}$ of type $(\tp(y))^{\dag_d}$.
	Likewise, for every nonterminal $X$ of $\Gg$ and for every function $A\in D_d^{\scope{\gar(X)}}$ we consider a nonterminal $X_A^{\dag_d}$ of type $(\tp(X))^{\dag_d}$.
	As the new set of nonterminals we take $\Xx^{\dag_d}=\setof{X_A^{\dag_d}}{X\in\Xx,A\in D_d^{\scope{\gar(X)}}}\cup\set{\upVdash,\downVdash}$.

	We now define a function $\tr_d$ transforming terms.
	Its value $\tr_d(A,Z,M)$ is defined when $M$ is a term over $(\Xx,\Yy,\Sigma_d)$ for some set of variables $\Yy$, and $A\in D_d^{\scope{\gar(M)}}$,
	and $Z\colon\Yy\rightharpoonup D_d$ is a partial function such that $\dom(Z)$ contains only variables of type $\otyp$.
	The intention is that $A$ specifies Eve's declarations for trailing order-$0$ arguments, and $Z$ specifies them for order-$0$ variables (among those in $\dom(Z)$).
	The transformation is defined by induction on the structure of $M$, as follows:
	\begin{bracketenumerate}
	\item	$\tr_d(A,Z,X)=X_A^{\dag_d}$ for $X\in\Xx$;
	\item	$\tr_d(A,Z,y)=y_A^{\dag_d}$ for $y\in\Yy\setminus\dom(Z)$;
	\item\label[case-br]{tr:case:3}
		$\tr_d(\emptyset,Z,z)=\downVdash$ if $Z(z)$ is odd;
	\item\label[case-br]{tr:case:4}
		$\tr_d(\emptyset,Z,z)=\upVdash$ if $Z(z)$ is even;
	\item\label[case-br]{tr:case:5}
		$\tr_d(\emptyset,Z,\symb{\player,p,K_1,\dots,K_k})=\symb{\player,p,\tr_d(\emptyset,Z\restr_p,K_1),\dots,\tr_d(\emptyset,Z\restr_p,K_k)}$,
		where $Z\restr_p$ is the function defined by $Z\restr_p(z)=(Z(z))\restr_p$ for all $z\in\dom(Z)$;
	\item\label[case-br]{tr:case:6}
		$\tr_d(A,Z,K\,L)=\symb{\Eve,1,K_1^L,K_2^L,\dots,K_d^L,K_{2d}}$ if $\tp(K)=(\otyp^{\ell+1}\arr\otyp)$,
		where $K_r^L=\symb{\Adam,1,K_r,\symb{\Eve,r,\allowbreak\tr_d(\emptyset,Z\restr_r,L)}}$ for $r\in\scope{d}$
		and $K_r=\tr_d(A\mapch{\ell+1\mapsto r},Z,K)$ for $r\in D_d$;
	\item	$\tr_d(A,Z,K\,L)=(\tr_d(A,Z,K))\,(\tr_d(B,Z,L))_{B\in D_d^{\scope{\gar(L)}}}$ if $\tp(K)=(\alpha_1\arr\dots\arr\alpha_k\Arr\otyp^\ell\arr\otyp)$ with $k\geq 1$.
	\end{bracketenumerate}
	In \cref{tr:case:3,,tr:case:4,tr:case:5} the term is of type $\otyp$, so the ``$A$'' argument is necessarily $\emptyset$ (a function with an empty domain).

	For every rule $X\,y_1\,\dots\,y_k\,z_1\,\dots\,z_\ell\to R$ in $\Rr$, where $\ell=\gar(X)$,
	and for every function $A\in D_d^{\scope{\ell}}$,
	to $\Rr^{\dag_d}$ we take the rule
	\begin{align*}
		X_A^{\dag_d}\,(y_{1,B}^{\dag_d})_{B\in D_d^{\scope{\gar(y_1)}}}\,\dots\,(y_{k,B}^{\dag_d})_{B\in D_d^{\scope{\gar(y_k)}}}\to\tr_d(\emptyset,\mapch{z_i\mapsto A(\ell+1-i)\mid i\in\scope{\ell}},R).
	\end{align*}
	In the function $A$ it is more convenient to count arguments from right to left (then we do not need to shift the domain in \cref{tr:case:6} above),
	but it is more natural to have variables $z_1,\dots,z_\ell$ numbered from left to right;
	this is why in the rule for $X_A^{\dag_d}$ we assign to $z_i$ the value $A(\ell+1-i)$, not $A(i)$.
	Additionally, in $\Rr^{\dag_d}$ we have rules $\upVdash\to\symb{\Eve,1,\upVdash}$ and $\downVdash\to\symb{\Eve,2,\downVdash}$.
	Then Eve loses (wins) in the tree $\bot$ ($\top$) generated by $\Gg^\dag$ from $\upVdash$ ($\downVdash$, respectively).

	Finally, the resulting recursion scheme $\Gg^\dag$ is $(\Xx^{\dag_d},X_{0,\emptyset}^{\dag_d},\Sigma_d,\Rr^{\dag_d})$.
	This finishes the definition of the transformation.
	In the next \lcnamecref{sec:complexity} we analyze its complexity, and in \cref{sec:correctness} we justify its correctness.

	\begin{remark}
		Let us briefly compare our transformation with a transformation by Broadbent et al.\@ \cite{collapsible-arxiv} reducing the order of a collapsible pushdown automaton by one
		while preserving the winner of the generated parity game.
		Although their transformation seems technically more complicated, its overall idea is quite similar to what we do in this paper.
		Their transformation is split into three independent steps.
		First, they make the automaton ``rank-aware'', which means that it knows what was the highest priority visited between creation of a collapse link and its usage.
		This corresponds to adding the parameters $A$ and $Z$ to our transformation, so that we know whether a declaration is fulfilled when a variable $z$ is used.
		Second, they eliminate collapse links of order $n$, which in our case corresponds to removing trailing arguments of order $0$
		and introducing the gadget asking Eve for a declaration.
		Third, they reduce the order of the automaton by one, which we also do for recursion schemes.
	\end{remark}

\section{Complexity}\label{sec:complexity}

	In this section we analyze complexity of our transformation.
	First, we formally define the \emph{size} of a recursion scheme.
	The size of a term is defined by induction on its structure:
	\begin{gather*}
		|X|=|y|=1,\qquad
		|K\,L|=1+|K|+|L|,\\
		|\symb{a,K_1,\dots,K_k}|=1+|K_1|+\dots+|K_k|.
	\end{gather*}
	Then $|\Gg|$, the size of $\Gg$, is defined as the sum of $|R|+k$ over all rules $X\,y_1\,\dots\,y_k\to R$ of $\Gg$.
	In Asada and Kobayashi~\cite{word2tree} such a size is called \emph{Curry-style} size;
	it does not include sizes of types of employed variables.

	We say that a type $\alpha$ \emph{appears in the definition} of a type $\beta$ if either $\alpha=\beta$,
	or $\beta=(\beta_1\arr\beta_2)$ and $\alpha$ appears in the definition of $\beta_1$ or of $\beta_2$.
	We write $A_\Gg$ for the largest arity of types appearing in the definition of types of nonterminals in a recursion scheme $\Gg$.
	Notice that types of other objects used in $\Gg$, namely variables and subterms of right-hand sides of rules, appear in the definition of types of nonterminals,
	hence their arity is also bounded by $A_\Gg$.
	It is reasonable to consider large recursion schemes, consisting of many rules, where simultaneously the maximal arity $A_\Gg$ is respectively small.

	While the exponential bound mentioned in \cref{thm:main} is obtained by applying the order-reducing transformation to an arbitrary parity recursion scheme,
	the complexity becomes slightly better if we first apply a preprocessing step.
	This is in particular necessary, if we want to obtain linear dependence in the size of $\Gg$ (and exponential only in the maximal arity $A_\Gg$).
	The preprocessing, making sure that the recursion scheme is in a \emph{simple form} (defined below), amounts to splitting large rules into multiple smaller rules.
	A similar preprocessing is present already in prior work~\cite{Kobayashi-jacm,word2tree,diagonal-arxiv,trans-nonempty}.

	An \emph{application depth} of a term $R$ is defined as the maximal number of applications on a single branch in $R$,
	where a compound application $K\,L_1\,\dots\,L_k$ counts only once.
	More formally, we define by induction:
	\begin{align*}
		&\mathsf{ad}(\symb{a,K_1,\dots,K_k})=\max\set{\mathsf{ad}(K_i)\mid i\in\scope{k}},\\
		&\mathsf{ad}(X\,K_1\,\dots\,K_k)=\mathsf{ad}(y\,K_1\,\dots\,K_k)=\max(\set{0}\cup\set{\mathsf{ad}(K_i)+1\mid i\in\scope{k}}).
	\end{align*}
	We say that a recursion scheme $\Gg$ is in a \emph{simple form} if the right-hand side of each its rule has application depth at most $2$.
	We have the following:

	\begin{lemma}[{\cite[Lemma~4.1]{trans-nonempty}}]\label{simpl-complexity}
		For every recursion scheme $\Gg$ there exists a recursion scheme $\Gg'$ being in a simple form, generating the same tree as $\Gg$, and such that
		$\ord(\Gg')=\ord(\Gg)$, and $A_{\Gg'}\leq 2A_\Gg$, and $|\Gg'|=\Oo(A_\Gg\cdot|\Gg|)$.
		The recursion scheme $\Gg'$ can be created in time linear in its size.
	\end{lemma}
	
	We now state and prove the main lemma of this section:
	
	\begin{lemma}\label{trans-complexity}
		For every parity recursion scheme $\Gg=(\Xx,X_0,\Sigma_d,\Rr)$ in a simple form, the recursion scheme $\Gg^\dag$ (i.e., the result of the order-reducing transformation)
		is also in a simple form, and $\ord(\Gg^\dag)=\max(0,\ord(\Gg)-1)$, and $A_{\Gg^\dag}\leq A_\Gg\cdot (d+1)^{A_\Gg}$,
		and $|\Gg^\dag|=\Oo(|\Gg|\cdot (d+1)^{5\cdot A_\Gg})$.
		Moreover, $\Gg^\dag$ can be created in time linear in its size.
	\end{lemma}

	\begin{proof}
		The part about the running time is obvious.
		It is also easy to see by induction that $\ord(\alpha^{\dag_d})=\max(0,\ord(\alpha)-1)$.
		It follows that the order of the recursion scheme satisfies the same equality,
		because nonterminals of $\Gg^\dag$ have type $\alpha^{\dag_d}$ for $\alpha$ being the type of a corresponding nonterminal of $\Gg$.

		Recall that in the type $\alpha^{\dag_d}$ obtained from $\alpha=(\alpha_1\arr\dots\arr\alpha_k\arr\otyp)$,
		every $\alpha_i$ either disappears or becomes (transformed and) repeated $|D_d|^{\gar(\alpha_i)}$ times, that is, at most $(d+1)^{A_\Gg}$ times.
		This implies the inequality concerning $A_{\Gg^\dag}$.

		Every compound application can be written as $f\,K_1\,\dots\,K_k\,L_1\,\dots\,L_\ell$, where $f$ is a nonterminal or a variable, and $\ell=\gar(f)$.
		In such a term, every $K_i$ (after being transformed) gets repeated $|D_d|^{\gar(K_i)}$ times, that is, at most $(d+1)^{A_\Gg}$ times.
		Then, for every $L_i$ we replicate the outcome $d+1$ times, and we append a small prefix;
		this replication happens $\ell$ times (and $\ell\leq A_\Gg$).
		In consequence, we easily see by induction that while transforming a term of application depth $c$, its size gets multiplicated by at most $O((d+1)^{2c\cdot A_\Gg})$.
		Moreover, every nonterminal $X$ is repeated $|D_d|^{\gar(X)}$ times, that is, at most $(d+1)^{A_\Gg}$ times.
		Because the application depth of right-hand sides of rules is at most $2$, this bounds the size of the new recursion scheme by $\Oo(|\Gg|\cdot (d+1)^{5\cdot A_\Gg})$.

		Looking again at the above description of the transformation, we can notice that the application depth cannot grow;
		in consequence the property of being in a simple form is preserved.
	\end{proof}

	Thus, if we want to check whether Eve wins in the tree generated by a parity recursion scheme $\Gg$ of order $n$,
	we can first convert $\Gg$ to a simple form, and then apply the order-reducing transformation $n$ times.
	This gives us a parity recursion scheme of order $0$, which can be seen as a finite parity game with $d$ priorities.
	Such a game can be solved in time $O(N^4\cdot 2^d)$, where $N$ is its size~\cite{Calude}.
	Thus, by \cref{simpl-complexity,trans-complexity}, the whole algorithm works in time $n$-fold exponential in $A_\Gg$ and $d$, and polynomial (quartic) in $|\Gg|$.
	
	If $\Gg$ is created as a product of a recursion scheme $\Hh$ and an alternating parity automaton $\Aa$,
	the running time is $n$-fold exponential in $A_\Hh$ and $|\Aa|$, and quartic in $|\Hh|$ (cf.~\cref{app:product}).

\section{Correctness}\label{sec:correctness}

	In this section we finish a proof of \cref{thm:main} by showing that the winner in the tree generated by the recursion scheme $\Gg^\dag$ resulting from transforming a recursion scheme $\Gg$
	is the same as in the tree generated by the original recursion scheme $\Gg$.
	Our proof consists of three parts.
	First, we show that reductions performed by $\Gg$ can be reordered, so that we can postpone substituting for (trailing) variables of order $0$.
	To store such postponed substitutions, called \emph{explicit substitutions}, we introduce \emph{extended trees}.
	Second, we show that such reordered reductions in $\Gg$ are in a direct correspondence with reductions in $\Gg^\dag$.
	Finally, we show how winning strategies of particular players from the tree generated by $\Gg^\dag$ can be transferred to the tree generated by $\Gg$.

    \subparagraph{Extended trees and terms.}

	In the sequel, trees and terms defined previously are sometimes called non-extended trees and non-extended terms,
	in order to distinguish them from extended trees and extended terms defined below.
	Having a set $\Zz$ of variables of type $\otyp$ and a set of symbols $\Sigma$,
	(potentially infinite) \emph{extended trees} over $(\Zz,\Sigma)$ are defined by coinduction: every extended tree over $(\Zz,\Sigma)$ is of the form either
	\begin{itemize}
	\item	$\symb{a,T_1,\dots,T_k}$, where $a\in\Sigma$ and $T_1,\dots,T_k$ are again extended trees over $\Sigma$, or
	\item	$z$ for some variable $z\in\Zz$, or
	\item	$\esubst{T}{U}{z}$, where $z\not\in\Zz$ is a variable of type $\otyp$, and $T$ is an extended tree over $(\Zz\cup\set{z},\Sigma)$,
		and $U$ is an extended tree over $(\Zz,\Sigma)$.
	\end{itemize}
	The construction of the form $\esubst{T}{U}{z}$ is called an \emph{explicit substitution}.
	Intuitively, it denotes the tree obtained by substituting $U$ for $z$ in $T$.
	Notice that the variable $z$ being free in $T$ becomes bound in $\esubst{T}{U}{z}$.

	Likewise, having a set of typed nonterminals $\Xx$, a set $\Zz$ of variables of type $\otyp$, and a set of symbols $\Sigma$,
	\emph{extended terms} over $(\Xx,\Zz,\Sigma)$ are defined by induction:
	\begin{itemize}
	\item	if $z\not\in\Zz$ is a variable of type $\otyp$, and $E$ is an extended term over $(\Xx,\Zz\cup\set{z},\Sigma)$, and $L$ is a non-extended term of type $\otyp$ over $(\Xx,\Zz,\Sigma)$,
		then $\esubst{E}{L}{z}$ is an extended term over $(\Xx,\Zz,\Sigma)$;
	\item	every non-extended term of type $\otyp$ over $(\Xx,\Zz,\Sigma)$ is an extended term over $(\Xx,\Zz,\Sigma)$.
	\end{itemize}
	Notice that explicit substitutions can be placed anywhere inside an extended tree, while in an extended term they are allowed only to surround a non-extended term.

	Of course an extended tree over $(\Zz,\Sigma)$ can be also seen as an extended tree over $(\Zz',\Sigma)$, where $\Zz'\supseteq\Zz$;
	likewise for extended terms.
	In the sequel, such extending of the set of variables is often performed implicitly.

	Having a recursion scheme $\Gg=(\Xx,X_0,\Sigma,\Rr)$, for every set $\Zz$ of variables of type $\otyp$
	we define an \emph{ext-reduction} relation $\erew_\Gg$ between extended terms over $(\Xx,\Zz,\Sigma)$,
	as the least relation such that
	\begin{itemize}
	\item	$X\,K_1\,\dots\,K_k\,L_1\,\dots\,L_\ell\erew_\Gg\esubstdots{R\subst{K_1/y_1,\dots,K_k/y_k,z_1'/z_1,\dots,z_\ell'/z_\ell}}{L_1}{z_1'}{L_\ell}{z_\ell'}$
		if\linebreak $\ell=\gar(X)$, and $\Rr(X)=(X\,y_1\,\dots\,y_k\,z_1\,\dots\,z_\ell\to R)$, and $z_1',\dots,z_\ell'$ are fresh variables of type $\otyp$ not appearing in $\Zz$.
	\end{itemize}
	Then, we define by coinduction the extended tree (over $(\Zz,\Sigma)$) \emph{ext-generated} by $\Gg$ from an extended term $E$ (over $(\Xx,\Zz,\Sigma)$), denoted $\BText_\Gg(E)$:
	\begin{itemize}
	\item	if $E\erew_\Gg^*\symb{a,F_1,\dots,F_k}$, then $\BText_\Gg(E)=\symb{a,\BText_\Gg(F_1),\dots,\BText_\Gg(F_k)}$;
	\item	if $E\erew_\Gg^*\esubst{F}{L}{z}$, then $\BText_\Gg(E)=\esubst{\BText_\Gg(F)}{\BText_\Gg(L)}{z}$;
	\item	otherwise, $\BText_\Gg(E)=\symb{\omega}$.
	\end{itemize}
	The extended tree ext-generated by $\Gg$ (without mentioning a term), denoted $\BText(\Gg)$, is defined as $\BText_\Gg(X_0)$.
	Formally, the ext-generated extended tree is not unique, because arbitrary fresh names may be used for bound variables;
	we should thus identify extended trees differing only in names of bound variables.

	Finally, we say how to convert extended trees to trees, by performing all postponed substitutions.
	To this end, having fixed a set $\Sigma$ of symbols,
	we define a \emph{simplification} relation $\simpl$ between extended trees over $(\emptyset,\Sigma)$
	as the least relation such that
	\begin{itemize}
	\item	$\esubstdots{\symb{a,T_1,\dots,T_k}}{L_1}{z_1}{L_\ell}{z_\ell}\simpl\symb{a,\esubstdots{T_1}{L_1}{z_1}{L_\ell}{z_\ell},\dots,\esubstdots{T_k}{L_1}{z_1}{L_\ell}{z_\ell}}$, and
	\item	$\esubstdots{z_i}{L_1}{z_1}{L_\ell}{z_\ell}\simpl\esubstdots{L_i}{L_{i+1}}{z_{i+1}}{L_\ell}{z_\ell}$.
	\end{itemize}
	Then, we define by coinduction the \emph{expansion} of an extended tree $T$ over $(\emptyset,\Sigma)$, being a tree over $\Sigma$, and denoted $\BTsimpl(T)$:
	\begin{itemize}
	\item	if $T\simpl^*\symb{a,T_1,\dots,T_k}$, then $\BTsimpl(T)=\symb{a,\BTsimpl(T_1),\dots,\BTsimpl(T_k)}$;
	\item	otherwise, $\BTsimpl(T)=\symb{\omega}$.
	\end{itemize}

	The following \lcnamecref{std2ext} says that instead of generating a tree, we can first ext-generate an extended tree, and then expand all the explicit substitutions:

	\begin{lemma}\label{std2ext}
		For every recursion scheme $\Gg$ it holds that $\BT(\Gg)=\BTsimpl(\BText(\Gg))$.
	\end{lemma}

	The lemma can be proved in a standard way;
	a proof is contained in \cref{app:std2ext} (similar lemmata appear in previous work~\cite[Lemma 18]{word2tree},~\cite[Lemma 5.1]{trans-nonempty}).

    \subparagraph{Transforming extended parity trees.}

    	An \emph{extended parity tree} is an extended tree whose expansion is a parity tree.
    	We now show how the transformation, defined previously for terms, can be applied to extended parity trees.
    	Namely, we define $\trt_d(Z,T)$ when $T$ is an extended tree over $(\Zz,\Sigma_d)$ for some set $\Zz$ of variables of type $\otyp$,
    	and $Z\colon\Zz\to D_d$ (we do not need an ``$A$'' argument, used previously to store declarations for arguments, because extended trees have no arguments).
    	The definition is by coinduction:
	\begin{enumerate}[(1')]
	\setcounter{enumi}{2}
	\item	$\trt_d(Z,z)=\top$ if $Z(z)$ is odd;
	\item	$\trt_d(Z,z)=\bot$ if $Z(z)$ is even;
	\item	$\trt_d(Z,\symb{\player,p,K_1,\dots,K_k})=\symb{\player,p,\trt_d(Z\restr_p,K_1),\dots,\trt_d(Z\restr_p,K_k)}$;
	\setcounter{enumi}{7}
	\item	$\trt_d(Z,\esubst{T}{U}{z})=\symb{\Eve,1,T_1^U,T_2^U,\dots,T_d^U,T_{2d}}$,
		where we take $T_r^U=\symb{\Adam,1,T_r,\symb{\Eve,\allowbreak r,\allowbreak\trt_d(Z\restr_r,\allowbreak U)}}$ for $r\in\scope{d}$
		and $T_r=\trt_d(Z\mapch{z\mapsto r},T)$ for $r\in D_d$.
	\end{enumerate}
	Notice that $\tr_d$ transforms a term $z$ to nonterminals $\downVdash$ or $\upVdash$,
	while $\trt_d$ transforms an extended tree $z$ to trees $\top$ or $\bot$, generated from those nonterminals.
	
	In the next \lcnamecref{ext2trans} we observe that the tree generated by the transformed recursion scheme $\Gg^\dag$
	can be obtained by transforming the extended tree ext-generated by the original recursion scheme $\Gg$:

	\begin{lemma}\label{ext2trans}
		For every parity recursion scheme $\Gg$ it holds that $\trt_d(\emptyset,\BText(\Gg))=\BT(\Gg^\dag)$.
	\end{lemma}

	The proof is purely syntactical, and is contained in \cref{app:ext2trans}.

    \subparagraph{Transforming strategies.}

	We finish our correctness proof by showing the following \lcnamecref{strat-trans}:
	
	\begin{lemma}\label{strat-trans}
		Let $T$ be an extended parity tree over $(\emptyset,\Sigma_d)$.
		If a player $\player\in\set{\Adam,\Eve}$ wins in $\trt_d(\emptyset,T)$, then $\player$ wins also in $\BTsimpl(T)$.
	\end{lemma}
	
	Recall that the goal of this section is to prove that the winner in $\BT(\Gg^\dag)$ is the same as in $\BT(\Gg)$, for every parity recursion scheme $\Gg$.
	This follows from the above \lcnamecref{strat-trans} used for $T=\BText(\Gg)$, because $\BT(\Gg^\dag)=\trt_d(\emptyset,\BText(\Gg))$ by \cref{ext2trans}
	and $\BT(\Gg)=\BTsimpl(\BText(\Gg))$ by \cref{std2ext}.
	
	We now come to a proof of \cref{strat-trans}.
	In the sequel we assume a fixed extended parity tree $T$ over $(\emptyset,\Sigma_d)$.
	Suppose first that it is Eve who wins in $\trt_d(\emptyset,T)$;
	thus, we also fix her winning strategy $\rho$ in this tree.
	Our goal is to construct Eve's winning strategy $\rho'$ in $\BTsimpl(T)$.
	
	In the proof, we use two additional notions.
	First, we say that a sequence of priorities $r_1,\dots,r_k$ is a \emph{$\preceq$-contraction} of a sequence of priorities $p_1,\dots,p_n$
	if the latter can be split at some indices $i_0,i_1,\dots,i_k$,
	where $0=i_0\leq i_1\leq\dots\leq i_k=n$, so that
	for every $j\in\scope{k}$ the infix $p_{i_{j-1}+1},p_{i_{j-1}+2},\dots,p_{i_j}$ fulfils declaration $r_j$.
	Likewise we define $\preceq$-contractions for infinite sequences, only there are infinitely many splitting indices
	(which necessarily tend to infinity, meaning that the whole infinite sequence is split).

	Notice that we allow empty infixes, so one can arbitrarily insert odd numbers $r_j$ (i.e., numbers $r_j$ fulfilled by the empty sequence) to the $\preceq$-contraction.
	For example, $3,4,2$ is a $\preceq$-contraction of $4,3,2,3,4$ because the empty sequence fulfils $3$, and $4,3$ fulfils $4$, and $2,3,4$ fulfils $2$.
	On the other hand, $3,4,2$ is not a $\preceq$-contraction of $4,3,2,3$.
	The idea of $\preceq$-contractions is to describe what happens when we move from $\BTsimpl(T)$ to $\trt_d(\emptyset,T)$.
	Indeed, if $T$ has a subtree of the form $\esubst{U}{V}{z}$,
	then in $\trt_d(\emptyset,T)$ the play can continue to $V$ after playing only an Eve's declaration $r$ (skipping completely $U$),
	while in $\BTsimpl(T)$ before reaching $V$ we traverse through $U$, where visited priorities are intended to fulfil $r$.

	It is easy to see that $\preceq$-contractions are transitive, and that they can make the situation only worse for Eve:
	
	\begin{lemma}\label{contraction-transitive}
		If a sequence $\pi_1$ is a $\preceq$-contraction of a sequence $\pi_2$, which is in turn a $\preceq$-contraction of a sequence $\pi_3$,
		then $\pi_1$ is a $\preceq$-contraction of $\pi_3$.
	\end{lemma}

	\begin{lemma}\label{contsctions-preserves-win}
		If an infinite sequence $\pi_1$ is a $\preceq$-contraction of an infinite sequence $\pi_2$,
		and the greatest priority appearing infinitely often in $\pi_1$ is even, then the greatest priority appearing infinitely often in $\pi_2$ is even as well.
	\end{lemma}
	
	We now introduce the second notion (it concerns only finite sequences, and is relative to the bound $d$ on priorities):
	for a declaration $r\in D_d$ and two sequences $\pi_1,\pi_2$ of priorities from $\scope{d}$ we say that $\pi_1$ is an \emph{$r$-extension} of $\pi_2$
	if for every sequence $\pi_3$ of priorities from $\scope{d}$ that fulfils the declaration $r$,
	the sequence $\pi_1$ is a $\preceq$-contraction of the concatenation $\pi_2\cdot\pi_3$.

	For example, the sequence $3,4,4$ is a $5$-extension of the sequence $4,3,6$ (independently from the value of $d\geq 6$),
	because the empty sequence fulfils $3$, and $4,3$ fulfils $4$, and $6,p_1,\dots,p_k$
	fulfils $4$ whenever $p_1,\dots,p_k$ fulfils $5$ (i.e., the maximum among $p_1,\dots,p_k$ is either even or at most $5$).
	Notice, moreover, that every sequence is a $2d$-extension of every sequence, because no sequence of priorities from $\scope{d}$ can fulfil the declaration $2d$.
	
	The following \lcnamecref{extension-shifted} is a direct consequence of the definition and of \cref{shift2fulfilled}:

	\begin{lemma}\label{extension-shifted}
		If a sequence $\pi$ is an $r$-extension of a sequence $p_1,\dots,p_n$,
		then $\pi$ is also an $r\restr_{p_{n+1}}$-extension of $p_1,\dots,p_n,p_{n+1}$ for every priority $p_{n+1}\in\scope{d}$.
	\end{lemma}
	
	Additionally, for a node $v$ (of some parity tree) we write $\pi(v)$ for the sequence of priorities in ancestors of $v$ (not including the priority in $v$).

	We now come back to the proof, showing how to construct the new strategy $\rho'$, winning for Eve in $\BTsimpl(T)$.
	In order to describe $\rho'$, we play simultaneously in both trees, $\BTsimpl(T)$ and $\trt_d(\emptyset,T)$,
	and we use moves in one tree to choose moves in the other tree.
	Namely, at every moment of the play, we remember
	\begin{itemize}
	\item	a current node $v$ in $\BTsimpl(T)$,
	\item	nodes $w_0,w_1,\dots,w_\ell$ in $\trt_d(\emptyset,T)$, for some $\ell\in\Nat$,
	\item	variables $z_1,\dots,z_\ell$ of type $\otyp$,
	\item	functions $Z_0,Z_1,\dots,Z_\ell$ storing Eve's declarations, where $Z_i\colon\set{z_{i+1},\dots,z_\ell}\to D_d$ for every $i$, and
	\item	extended trees $U_0,U_1,\dots,U_\ell$, where every $U_i$ is over $(\set{z_{i+1},\dots,z_\ell},\Sigma_d)$.
	\end{itemize}
	They satisfy the following invariant:
	\begin{enumerate}[(a)]
	\item	$\BTsimpl(T)\subtree_v=\BTsimpl(\esubstdots{U_0}{U_1}{z_1}{U_\ell}{z_\ell})$,
	\item	$\trt_d(\emptyset,T)\subtree_{w_i}=\trt_d(Z_i,U_i)$ for all $i\in\set{0,1,\dots,\ell}$,
	\item\label{inv:c}
		$\pi(w_0)$ is a $\preceq$-contraction of $\pi(v)$, and
	\item\label{inv:d}
		$\pi(w_j)$ is a $Z_i(z_j)$-extension of $\pi(w_i)$,
		for all $i,j$ such that $0\leq i<j\leq\ell$.
	\end{enumerate}
	
	We start with $\ell=0$, with $v$ and $w_0$ at the root of $\BTsimpl(T)$ and $\trt_d(\emptyset,T)$, respectively, with $Z_0=\emptyset$, and with $U_0=T$.
	The invariant is clearly satisfied.
	
	Then, during the play, we have one of three cases, depending on the shape of $U_0$:
	\begin{enumerate}
	\item\label[case]{case:1}
		First, assume that $U_0=\symb{\player,p,T_1,\dots,T_k}$.
		Then
		\begin{align*}
			\BTsimpl(T)\subtree_v&=\symb{\player,p,\BTsimpl(\esubstdots{T_1}{U_1}{z_1}{U_\ell}{z_\ell}),\dots,\BTsimpl(\esubstdots{T_k}{U_1}{z_1}{U_\ell}{z_\ell})};\\
			\trt_d(\emptyset,T)\subtree_{w_0}&=\symb{\player,p,\trt_d(Z_0,T_1),\dots,\trt_d(Z_0,T_k)}.
		\end{align*}
		If $\player=\Adam$, Adam chooses some child of $v$ in $\BTsimpl(T)$,
		and we choose the same child of $w_0$ in $\trt_d(\emptyset,T)$.
		If $\player=\Eve$, Eve chooses some child of $w_0$ in $\trt_d(\emptyset,T)$, according to her strategy $\rho$, and in $\rho'$ we choose the same child of $v$.
		Thus, in both cases, we move both $v$ and $w_0$ to their $c$-th child, for some $c\in\scope{k}$.
		We also take $Z_0\restr_p$ as the new $Z_0$ and $T_c$ as the new $U_0$.
		\cref{extension-shifted} ensures that \cref{inv:d} of the invariant is preserved.
	\item\label[case]{case:2}
		Another possibility is that $U_0$ is a variable, that is, $U_0=z_c$ for some $c\in\scope{\ell}$.
		Then $\trt_d(\emptyset,T)\subtree_{w_0}$ (i.e., $\trt_d(Z_0,U_0)$) is either $\bot$ or $\top$, depending on the parity of $Z_0(z_c)$.
		But our play in $\trt_d(\emptyset,T)$ follows an Eve's winning strategy, so it will be won by Eve, thus the subtree cannot be $\bot$, in which Eve is losing.
		In consequence $Z_0(z_c)$ is odd, so the empty sequence fulfils $Z_0(z_c)$.
		This implies that $\pi(w_c)$, being an $Z_0(z_c)$-extension of $\pi(w_0)$, is its $\preceq$-contraction,
		and thus also an $\preceq$-contraction of $\pi(v)$ (by \cref{contraction-transitive}).
		We discard $w_i,z_i,Z_i,U_i$ for $i<c$ (so that $w_c$ becomes now $w_0$, etc.).
	\item\label[case]{case:3}
		Finally, assume that $U_0=\esubst{V}{W}{z}$.
		Then $\trt_d(\emptyset,T)\subtree_{w_0}=\symb{\Eve,1,V_1^W,\dots,V_d^W,V_{2d}}$,
		where $V_r^W=\symb{\Adam,1,V_r,\symb{\Eve,r,\trt_d(Z_0\restr_r,W)}}$ for $r\in\scope{d}$ and $V_r=\trt_d(Z_0\subst{z\mapsto r},V)$ for $r\in D_d$.
		In such a node $w_0$ Eve, according to her strategy $\rho$, chooses a declaration $r$ by going to an appropriate subtree $V_r^W$ (or $V_r$ if $r=2d$).
		We then update our memory as follows:
		\begin{itemize}
		\item	We leave $v$ and $w_i,z_i,Z_i,U_i$ for $i\geq 1$ unchanged.
		\item	We move $w_0$ to the root of $V_r$
			(this adds once or twice priority $1$ to $\pi(w_0)$, hence \cref{inv:c} of the invariant is preserved).
		\item	Let $r'=r$ if $r\in\scope{d}$, and $r'=1$ if $r=2d$.
		\item	We add an additional node $w_{0.5}$ between $w_0$ and $w_1$ (saying this differently, we shift $w_i$ for $i\geq 1$ by one, and we insert the new node in place of $w_1$).
			For $w_{0.5}$ we choose the root of $\trt_d(Z_0\restr_{r'},W)$.
			Notice that $\pi(w_{0.5})$ is an $r$-extension of $\pi(w_0)$
			(for $r\in\scope{d}$ because $\pi(w_{0.5})$ is obtained from $\pi(w_0)$ by appending the priority $r'=r$,
			and for $r=2d$ because no sequence of priorities from $\scope{d}$ fulfils $2d$),
			and that every $\pi(w_j)$ for $1\leq j\leq\ell$ is a $Z_0\restr_{r'}(z_j)$-extension of $\pi(w_{0.5})$ (by \cref{extension-shifted}).
		\item	As $Z_0$, $U_0$, $z_{0.5}$, $Z_{0.5}$, and $U_{0.5}$ we take $Z_0\subst{z\mapsto r}$, $V$, $z$, $Z_0\restr_{r'}$, and $W$, respectively.
		\end{itemize}
	\end{enumerate}
	
	Observe that after finitely many repetitions of \cref{case:2,case:3} necessarily \cref{case:1} has to occur, where the play advances in $\BTsimpl(T)$.
	Indeed, $\esubstdots{U_0}{U_1}{z_1}{U_\ell}{z_\ell}$ has to generate the next node of $\BTsimpl(T)$ in finitely many steps;
	in particular, the number of explicit substitution at the head of $U_0$ has to be finite.
	
	We have to prove that the infinite branch $\xi$ of $\BTsimpl(T)$ obtained this way is won by Eve.
	To this end, consider the corresponding sequence of ``$w_0$'' nodes in the construction
	and observe that this sequence converges to some infinite branch $\zeta$ in $\trt_d(\emptyset,T)$.
	Indeed, whenever the sequence enters to a subtree of the form $\trt_d(Z_0,\esubst{V}{W}{z})$ and stays there forever,
	then either it enters to the subtree $V_r=\trt_d(Z_0\subst{z\mapsto r},V)$ for some $r$ and stays there forever,
	or, after some time, it enters to the subtree $\trt_d(Z_0\restr_r,W)$ for some $r$ and stays there forever.
	Moreover, the sequence of priorities on $\zeta$ is a $\preceq$-contraction of the sequence of priorities on $\xi$
	(the function from elements of the former sequence to infixes of the latter sequence, as needed for $\preceq$-contraction, is obtained as the limit
	of such functions witnessing that always $\pi(w_0)$ is a $\preceq$-contraction of $\pi(v)$).
	Since $\zeta$ agrees with the strategy $\rho$, it is won by Eve, hence by \cref{contsctions-preserves-win} also $\xi$ is won by Eve, as required.
	This finishes the proof in the case of Eve winning in $\trt_d(\emptyset,T)$.

	Suppose now that it is Adam who wins in $\trt_d(\emptyset,T)$.
	The proof in this case is similar, so we only list differences.
	First, \emph{$\succeq$-contraction} is defined like $\preceq$-contraction,
	but for every infix $p_{i_{j-1}+1},p_{i_{j-1}+2},\dots,p_{i_j}$ in the split we require that
	$r_j$ is $\succeq$ (instead of $\preceq$) than the leader of the infix.
	Second, we say that a sequence $\pi_1$ of priorities from $\scope{d}$ is an \emph{$r$-neg-extension} of a sequence $\pi_2$ of priorities from $\scope{d}$
	if for every sequence $\pi_3$ of priorities from $\scope{d}$ that does NOT fulfil the declaration $r$,
	the sequence $\pi_1$ is a $\succeq$-contraction of the concatenation $\pi_2\cdot\pi_3$.
	In \cref{inv:c,inv:d} of the invariant we replace $\preceq$-contraction by $\succeq$-contraction, and $r$-extension by $r$-neg-extension.
	Then, in \cref{case:1} of the construction we only swap the role of Eve and Adam.
	In \cref{case:2} we now have that the play is won by Adam, so $Z_0(z_c)$ is even, that is, not fulfilled by the empty sequence;
	this implies that $\pi(w_c)$, being an $Z_0(z_c)$-neg-extension of $\pi(w_0)$, is also its $\succeq$-contraction.
	The main difference is in \cref{case:3}.
	For every $r\in\scope{d}$ we know Adam's decision in the root of $V_r^W$, according to his winning strategy.
	Take the worst $r\in\scope{d}$ such that in $V_r^W$ Adam goes to the left subtree, or $r=2d$ if he goes right everywhere;
	in both cases, Adam's strategy allows to enter $V_r$.
	Let also $s$ be the best among priorities that are worse than $r$; in $V_s^W$ Adam goes to the right subtree
	(if there are no priorities worse than $r$, we choose $s$ arbitrarily, e.g., $s=1$).
	Then as the new $w_0$ we take the root of $V_r$, and as $w_{0.5}$ we take the root of $\trt_d(Z\restr_s,W)$.
	Notice that $\pi(w_{0.5})$ is an $r$-neg-extension of $\pi(w_0)$: $s$ is better or equal than the leader of every sequence not fulfilling $r$
	(also when $r$ is the worst priority, because no such a sequence exists), which ensures that the invariant is preserved.

\section{Final remarks}

	We have presented a new, simple model-checking algorithm for higher-order recursion schemes.
	One may ask whether this algorithm can be used in practice.
	Of course the complexity $n$\textsf{-EXPTIME} for recursion schemes of order $n$ is unacceptably large
	(even if we take into account the fact that we are $n$-fold exponential only in the arity of types and in the size of an automaton, not in the size of a recursion scheme),
	but one has to recall that there exist tools solving the considered problem in such a complexity.
	The reason why these tools work is that the time spent by them on ``easy'' inputs is much smaller than the worst-case complexity (and many ``typical inputs'' are indeed easy).
	Unfortunately, this is not the case for our algorithm: the size of the recursion scheme resulting from our transformation is always large.
	Moreover, it seems unlikely that any simple analysis of the resulting recursion scheme (like removing useless nonterminals or some control flow analysis) may help in reducing its size.
	Indeed, one can see that if no nonterminals nor arguments were useless in the original recursion scheme,
	then also no nonterminals nor arguments are useless in the resulting recursion scheme.
	Thus, our algorithm is mainly of a theoretical interest.

	It seems feasible that a transformation similar to the one presented in this paper can be used to solve
	the simultaneous unboundedness problem (aka.\@ diagonal problem)~\cite{diagonal-arxiv} for recursion schemes.
	Developing such a transformation is a possible direction for further work.

\bibliography{bib}

\newpage\appendix

\section{Reducing model checking to parity games}\label{app:product}

	We present here the original model-checking problem,
	and we show how it reduces to the problem actually considered in the paper.
	This part contains a rather standard material, but we include it for completeness.
	We advice the reader to read \cref{sec:prelim,sec:complexity} before reading this material, because it depends on notions introduced thereof.
	
	There exist multiple equivalent definitions of alternating parity automata operating on trees.
	We use the following definition: An \emph{alternating parity automaton} is a tuple $\Aa=(\Sigma,\ell_{\max},Q,Q_E,q_I,\delta,\eta)$,
	where $\Sigma$ is a finite \emph{input alphabet}, $\ell_{\max}\in\Nat$ is the maximal arity of nodes in considered trees,
	$Q$ is a finite set of \emph{states}, $Q_E\subseteq Q$ defines \emph{existential states} (states in $Q\setminus Q_E$ are \emph{universal}),
	$q_I\in Q$ is an \emph{initial state}, $\delta$ is a \emph{transition function}, and $\eta\colon Q\times\Sigma\to\Nat_+$ defines \emph{priorities}.
	The transition function $\delta$ maps every triple $(q,a,\ell)\in Q\times\Sigma\times\set{0,1,\dots,\ell_{\max}}$ to a nonempty subset of $Q\times\set{0,1,\dots,\ell}$.
	Such a value defines the behaviour in state $q$ and in a node with label $a$ and $\ell$ children;
	a pair $(p,c)$ in the set $\delta(q,a,\ell)$ means that we can change the state to $p$ proceeding to the $c$-th child of the current node if $c\geq 1$,
	or staying in the current node if $c=0$.
	As the size of $\Aa$, denoted $|\Aa|$, we take the sum of $|\delta(q,a,\ell)|$ over all triples $(q,a,\ell)$.
	
	We define semantics of such an automaton using run trees.
	Let $\Aa=(\Sigma,\allowbreak\ell_{\max},\allowbreak Q,\allowbreak Q_E,\allowbreak q_I,\allowbreak\delta,\allowbreak\eta)$ be an alternating parity automaton.
	We define $\RT_\Aa(q,T)$ (``$\RT$'' stands for a \emph{run tree}) by coinduction,
	where $q\in Q$ is a state, and $T$ is a tree over $\Sigma$ such that every node of $T$ has at most $\ell_{\max}$ children:
	if $T=\symb{a,T_1,\dots,T_\ell}$ and $\delta(q,a,\ell)=\set{(p_1,c_1),\dots,(p_k,c_k)}$, then
	\begin{align*}
		\RT_\Aa(q,T)=\symb{\player,\eta(q,a),\RT_\Aa(p_1,T_{c_1}),\dots,\RT_\Aa(p_k,T_{c_k})},
	\end{align*}
	where as $T_0$ we take the whole $T$, and where $\player=\Eve$ if $q\in Q_E$ and $\player=\Adam$ otherwise.
	Notice that the run tree is a parity tree over the alphabet $\Sigma_d$, where $d$ is the maximal priority appearing in $\eta$.
	In particular every node of this tree has at least one child (if $\Aa$ reaches a leaf of $T$, then it has to loop there forever, generating an infinite branch in the run tree).
	We say that $\Aa$ \emph{accepts} $T$ if Eve wins in $\RT_\Aa(q_I,T)$.
	
	Suppose now that the input alphabet of $\Aa$ is $\Sigma_\omega=\Sigma\uplus\set{\omega}$,
	and consider a recursion scheme $\Gg$ whose output alphabet is $\Sigma$, and such that
	all node constructors appearing in $\Gg$ have arity at most $\ell_{\max}$.
	Recall the quantities $|\Gg|$ (size) and $A_\Gg$ (maximal arity of types) from \cref{sec:complexity}.
	Our goal is to create a ``product'' recursion scheme $\Gg_\Aa$, generating a tree that is won by Eve if and only if $\Aa$ accepts $\BT(\Gg)$.

	As a first step, we eliminate from $\BT(\Gg)$ nodes labeled by $\omega$ (recall that $\omega$ is a special symbol used to denote places where $\Gg$ diverges).
	To this end, we create a new recursion scheme, $\Gg'$, with an extended output alphabet $\Sigma_\varepsilon=\Sigma\uplus\set{\varepsilon}$.
	In $\Gg'$ we use the the same nonterminals as in $\Gg$,
	but we replace every rule $X\,y_1\,\dots\,y_k\to R$ by $X\,y_1\,\dots\,y_k\to \symb{\varepsilon,R}$.
	This way, after every $\rew_{\Gg'}$ step, a new $\varepsilon$-labeled node is generated.
	As a result, every $\omega$-labeled node in the generated tree is replaced by an infinite branch of $\varepsilon$-labeled nodes;
	additionally, a lot of $\varepsilon$-labeled nodes is inserted in different places.
	We also modify $\Aa$ to $\Aa'=(\Sigma_\varepsilon,\ell_{\max},Q,Q_E,q_I,\delta',\eta')$ working over such a modified tree.
	Let $Q_\mathsf{acc}$ be the set of those states $q\in Q$ for which $\RT_\Aa(q,\symb{\omega})$ is won by Eve
	(i.e., such that $\Aa$ accepts the single-node tree $\symb{\omega}$ from state $q$).
	It is standard to check whether $q\in Q_\mathsf{acc}$: it amounts to solving a finite parity game, and thus can be done in time exponential (actually, quasipolynomial)
	in the size of $\Aa$.
	Then, for every $q$ and $\ell$ we take
	\begin{align*}
		\delta'(q,a,\ell)&=\left\{\begin{array}{ll}
			\delta(q,a,\ell)&\mbox{if }a\in\Sigma,\\
			\set{(q,1)}&\mbox{if }a=\varepsilon\mbox{ and }\ell\geq 1,\\
			\set{(q,0)}&\mbox{if }a=\varepsilon\mbox{ and }\ell=0,
		\end{array}\right.
		&\mbox{and}\\
		\eta'(q,a)&=\left\{\begin{array}{ll}
			\eta(q,a)+2&\mbox{if }a\in\Sigma,\\
			2&\mbox{if }a=\varepsilon\mbox{ and }q\in Q_\mathsf{acc},\\
			1&\mbox{if }a=\varepsilon\mbox{ and }q\not\in Q_\mathsf{acc}.
		\end{array}\right.
	\end{align*}
	This way, the run tree $\RT_{\Aa'}(q_I,\BT(\Gg'))$ is quite similar to $\RT_{\Aa}(q_I,\BT(\Gg))$, except that
	\begin{itemize}
	\item	all priorities are increased by $2$ (so that they dominate over priorities $1$ and $2$ present in newly inserted nodes);
	\item	every subtree of the form $\RT_\Aa(q,\symb{\omega})$ is replaced by an infinite branch
		with nodes of priority $2$ (if $\RT_\Aa(q,\symb{\omega})$ was won by Eve) or $1$ (if $\RT_\Aa(q,\symb{\omega})$ was won by Adam);
	\item	some additional nodes of priority $1$ or $2$ are inserted in different places.
	\end{itemize}
	It is easy to see such a modification of the run tree does not change the winner,
	that is, $\Aa$ accepts $\BT(\Gg)$ if and only if $\Aa'$ accepts $\BT(\Gg')$.
	Moreover, we have $\ord(\Gg')=\ord(\Gg)$, and $A_{\Gg'}=A_\Gg$, and $|\Gg'|=\Oo(|\Gg|)$, and $|\Aa'|=\Oo(|\Aa|)$.
	
	Next, we change $\Gg'$ into an equivalent recursion scheme $\Gg''$ that is in a simple form (as in \cref{sec:complexity}).
	More precisely, we need it to be in a \emph{strong simple form}, defined as follows:
	for every rule $X\,y_1\,\dots\,y_k\to R$ of $\Gg''$,
	\begin{bracketenumerate}
	\item\label{ssf:1}
		$R$ has application depth at most $2$ (as defined in \cref{sec:complexity}), and
	\item\label{ssf:2}
		no proper subterm of $R$ starts with a node constructor (i.e., a node constructor may appear only at the outermost position).
	\end{bracketenumerate}
	
	\cref{ssf:2} can be ensured as follows:
	Consider a rule $X\,y_1\,\dots\,y_k\to R$, and a proper subterm $K$ of $R$ of the form $\symb{a,K_1,\dots,K_\ell}$.
	Then we replace the occurrence of $K$ with $Y\,y_1\,\dots\,y_k$ for a fresh nonterminal $Y$,
	and we add the rule $Y\,y_1\,\dots\,y_k\to K$.
	By repeating such a replacement for every ``bad'' node constructor in every rule, we clearly obtain an equivalent recursion scheme satisfying \cref{ssf:2}.
	The modification only multiplicates the size by a constant.
	The way of ensuring \cref{ssf:1} is described in prior work~\cite[Lemma~4.1]{trans-nonempty}, and it preserves \cref{ssf:2}.
	We have that $\ord(\Gg'')=\ord(\Gg)$, and $A_{\Gg''}\leq 2A_\Gg$, and $|\Gg''|=\Oo(A_\Gg\cdot|\Gg|)$.

	We now create the actual product of $\Gg''$ and $\Aa'$, denoted $\Gg_\Aa$.
	In the product, we replace terms of type $\alpha$ by terms of type $\alpha^\ddag$, which is defined by induction:
	\begin{align*}
		(\alpha_1\arr\dots\arr\alpha_k\arr\otyp)^\ddag = \left((\alpha_1^\ddag)^{|Q|}\arr\dots\arr(\alpha_k^\ddag)^{|Q|}\arr\otyp\right).
	\end{align*}
	Thus, we repeat each argument $|Q|$ types (and we modify it recursively).
	We fix some order on states in $Q$,
	and we write $K\,(L_q)_{q\in Q}$ for an application $K\,L_{q_1}\,\dots\,L_{q_{|Q|}}$, where $q_1,\dots,q_{|Q|}$ are all states from $Q$ listed in the fixed order.

	For every variable $y$ and for every state $q\in Q$ we consider a variable $y_q^\ddag$ of type $(\tp(y))^\ddag$.
	Likewise, for every nonterminal $X$ of $\Gg''$ and for every state $q\in Q$ as a nonterminal of $\Gg_\Aa$ we take $X_q^\ddag$ of type $(\tp(X))^\ddag$.

	We now define a function $\pr$ transforming terms;
	its value $\pr(q,M)$ is defined when $M$ is a term without node constructors, and $q\in Q$.
	We take
	\begin{itemize}
	\item	$\pr(q,X)=X_q^\ddag$;
	\item	$\pr(q,y)=y_q^\ddag$;
	\item	$\pr(q,K\,L)=\pr(q,K)\,(\pr(p,L))_{p\in Q}$.
	\end{itemize}
	Consider now a rule $X\,y_1\,\dots\,y_k\to R$ of $\Gg''$, and a state $q\in Q$.
	\begin{enumerate}
	\item\label[case]{rule-trans-1}
		If $R$ does not contain node constructors, as a rule for $X_q$ in $\Gg_\Aa$ we take
		\begin{align*}
			X_q\,(y_{1,p}^\ddag)_{p\in Q}\,\dots\,(y_{k,p}^\ddag)_{p\in Q}\to\pr(q,R).
		\end{align*}
	\item\label[case]{rule-trans-2}
		Since $\Gg''$ is in a strong simple form,
		the only remaining possibility is that the right-hand side $R$ is of the form $\symb{a,K_1,\dots,K_\ell}$, where $K_1,\dots,K_\ell$ are terms without node constructors.
		In $\Gg_\Aa$, instead of generating an $a$-labeled node, we should rather generate a fragment of the run tree of $\Aa'$ concerning this node.
		Let $\delta'(q,a,\ell)=\set{(p_1,c_1),\dots,(p_k,c_k)}$.
		Let also $\player=\Eve$ if $q\in Q_E$, and $\player=\Adam$ otherwise.
		Then, as a rule for $X_q$ in $\Gg_\Aa$ we take
		\begin{align*}
			X_q\,(y_{1,p}^\ddag)_{p\in Q}\,\dots\,(y_{k,p}^\ddag)_{p\in Q}\to\symb{\player,\eta'(q,a),L_1,\dots,L_k},
		\end{align*}
		where for every $i\in\scope{k}$ the term $L_i$ is defined as follows.
		If $c_i\geq 1$ (i.e., $\Aa$ descends to the $c_i$-th child of the $a$-labeled node in state $p_i$),
		we simply take $L_i=\pr(p_i,K_{c_i})$.
		Otherwise, when $c_i=0$ (i.e., when $\Aa$ remains in the $a$-labeled node),
		as $L_i$ we take $L_i=X_{p_i}\,(y_{1,p}^\ddag)_{p\in Q}\,\dots\,(y_{k,p}^\ddag)_{p\in Q}$,
		that is, we generate the run tree of $\Aa$ starting from the same $a$-labeled node, but now from the state $p_i$.
		It is important here that the considered node constructor is indeed on the outermost position, so that we can use the corresponding nonterminal to come back to it.
	\end{enumerate}

	It is not difficult to prove that $\Gg_\Aa$ generates the run tree $\RT_{\Aa'}(q_I,\BT(\Gg''))$.
	This run tree is won by Eve if and only if $\Aa'$ accepts $\BT(\Gg'')$, that is, if and only if $\Aa$ accepts $\BT(\Gg)$.

	Let us now determine the size of $\Gg_\Aa$.
	First, observe that the definition of the $\pr$ function for an application copies the argument $|Q|$ times.
	This copying is done again in the argument, which has a smaller application depth, so we can see that $|\pr(q,M)|\leq|M|\cdot|Q|^{\mathsf{ad}(M)}$.
	Because the application depth in $\Gg''$ is at most $2$, for \cref{rule-trans-1} above we have $|\pr(q,R)|\leq|R|\cdot|Q|^2$;
	we use $\pr(q,R)$ for every $q\in Q$, that is, $|Q|$ times.
	In \cref{rule-trans-2}, $\pr(p,K_c)$ for the largest subterm $K_c$ (or for some smaller subterm, instead) may be copied at most $|\delta'(q,a,\ell)|$ times in a single rule (i.e., for a single $q$),
	and we use it for all $q\in Q$, raising at most $\sum_{q\in Q}|\delta'(q,a,\ell)|\leq|\Aa'|$ copies.
	In both cases, arguments on the left-hand side are copied $|Q|^2$ times ($|Q|$ times in every of $|Q|$ copies of a rule).
	Altogether, we get $|\Gg_\Aa|=\Oo(|\Gg''|\cdot|\Aa'|^3)=\Oo(A_\Gg\cdot|\Gg|\cdot|\Aa|^3)$.
	The number of arguments in a type gets increased $|Q|$ times, so we have $A_{\Gg_\Aa}\leq A_{\Gg''}\cdot|Q|\leq 2A_\Gg\cdot|Q|$.
	The order does not change; $\ord(\Gg_\Aa)=\ord(\Gg)$.

\section{Proof of Lemma~\ref{shift2fulfilled}}\label{app:shift2fulfilled}

	\shiftfulfilled*

	\begin{proof}
		Let $p=p_1$, and let $s$ be the leader of $p_2,\dots,p_k$.
		Observe that the leader of $p_1,p_2,\dots,p_k$ is $\max(p,s)$.
		We have to prove that
		\begin{align*}
			r\preceq\max(p,s)\ \Leftrightarrow\ r\restr_p\preceq s.
		\end{align*}
		Recall the $\preceq$ order:
		\begin{align*}
			\cdots\preceq 5\preceq 3\preceq 1\preceq 2\preceq 4\preceq 6\preceq\cdots.
		\end{align*}
		We have the following cases:
		\begin{enumerate}
		\item	Suppose that $p$ is odd and $p>r$.
			Then, by definition, $r\restr_p=p+1$.
			Because $\max(p,s)>r$, we have $r\preceq\max(p,s)$ exactly when $\max(p,s)$ is even.
			This in turn holds when $s$ is even and $s\geq p+1$ (because $p$ is odd)
			and, by definition, can be expressed as $p+1\preceq s$, that is, $r\restr_p\preceq s$.
		\item	Suppose that $p$ is even and $p\geq r$.
			Then, by definition, $r\restr_p=p-1$.
			Because $\max(p,s)\geq r$, if $\max(p,s)$ is even then $r\preceq\max(p,s)$.
			Conversely, if $\max(p,s)$ is odd, then $\max(p,s)>p\geq r$ (since $p$ is even), so $r\succ\max(p,s)$.
			Thus, we have $r\preceq\max(p,s)$ exactly when $\max(p,s)$ is even.
			This in turn holds when either $s$ is even or ($s$ is odd and) $s\leq p-1$,
			and, by definition, can be expressed as $p-1\preceq s$, that is, $r\restr_p\preceq s$.
		\item	If none of the above holds, $r\restr_p=r$.
			\begin{enumerate}
			\item	If $s\geq p$, that is, $\max(p,s)=s$, both sides of the equivalence become $r\preceq s$.
			\item	Otherwise $s<p$, meaning that $\max(p,s)=p$;
				moreover $p\leq r$, and if $p$ is even then actually $p<r$.
				If $r$ is odd, then $r\preceq p$.
				Contrarily, if $r$ is even, then necessarily $p<r$, so $r\succ p$.
				This means that $r\preceq\max(p,s)$ exactly when $r$ is odd.
				We finish the proof observing that, due to $s<p\leq r$, we have $r\restr_p\preceq s$ (i.e., $r\preceq s$) exactly when $r$ is odd.
			\qedhere\end{enumerate}
		\end{enumerate}
	\end{proof}

\section{Proof of Lemma~\ref{std2ext}}\label{app:std2ext}

	In this section we prove \cref{std2ext}.
	We use here an ``expand'' function $\exp$ from extended terms to non-extended terms, which performs all the explicit substitutions written in front of an extended term:
	\begin{align*}
		\exp(\esubstdots{K}{L_1}{z_1}{L_\ell}{z_\ell})=K\subst{L_1/z_1}\dots\subst{L_\ell/z_\ell}.
	\end{align*}
	
	We now have a counterpart of \cref{std2ext} that is suitable for a (co)inductive proof;
	the original statement can be obtained by taking the starting nonterminal $X_0$ as $E$.

	\begin{lemma}\label{std2ext-aux}
		Let $\Gg=(\Xx,X_0,\Sigma,\Rr)$ be a recursion scheme.
		For every extended term $E$ over $(\Xx,\emptyset,\Sigma)$ it holds that $\BT_\Gg(\exp(E))=\BTsimpl(\BText_\Gg(E))$.
	\end{lemma}
	
	\begin{proof}
		Let us start by noting two simple facts, which are used implicitly below.
		First, for every term $M$ there is at most one term $N$ such that $M\rew_\Gg N$.
		Second, if $M\rew_\Gg N$, then $\BT_\Gg(M)=\BT_\Gg(N)$.
		Moreover, the same two facts hold for $\erew_\Gg$ and $\BText_\Gg$, and for $\simpl$ and $\BTsimpl$.
		
		The proof of the \lcnamecref{std2ext-aux} is by coinduction on the structure of $\BT_\Gg(\exp(E))$.
		Let us write $E=\esubstdots{M}{P_1}{z_1'}{P_s}{z_s'}$, where $M$ is a non-extended term.
		For any term $K$ denote $\Xi(K)=K\subst{P_1/z_1'}\dots\subst{P_s/z_s'}$ and for any extended term $F$ denote $\Upsilon(F)=\esubstdots{F}{P_1}{z_1'}{P_s}{z_s'}$.
		Then, in particular, $E=\Upsilon(M)$ and $\exp(\Upsilon(F))=\Xi(\exp(F))$; obviously also $\exp(M)=M$.
		
		We have two cases.
		The important case is when the root of $\BT_\Gg(\exp(E))$ is labeled by some $a\in\Sigma$ (i.e., not by $\omega$).
		Then $\exp(E)\rew_\Gg^r\symb{a,N_1,\dots,N_n}$ for some $r\in\Nat$ and $\BT_\Gg(\exp(E))=\symb{a,\BT_\Gg(N_1),\dots,\BT_\Gg(N_n)}$.
		We perform an internal induction on the pair $(r,s)$ (i.e., first on $r$, and internally on $s$),
		and we consider three subcases depending on the shape of $M$.
		
		The first subcase is when $M$ consists of a nonterminal $X$ to which some arguments are applied.
		We can write $M=X\,K_1\,\dots\,K_k\,L_1\,\dots\,L_\ell$, where $\ell=\gar(X)$.
		Let $X\,y_1\,\dots\,y_k\,z_1\,\dots\,z_\ell\allowbreak\to R$ be the rule for $X$.
		Denote $M'=R\subst{K_1/y_1,\dots,K_k/y_k,z'_{s+1}/z_1,\dots,z_{s+\ell}'/z_\ell}$ for fresh variables $z'_{s+1},\dots,z'_{s+\ell}$,
		and $F=\esubstdots{M'}{L_1}{z_{s+1}'}{L_\ell}{z_{s+\ell}'}$, and $E'=\Upsilon(F)$.
		On the one hand, by definition, $M\erew_\Gg F$, so
		\begin{align*}
			\BText_\Gg(E)&=\esubstdots{(\BText_\Gg(M))}{\BText_\Gg(P_1)}{z_1'}{\BText_\Gg(P_s)}{z_s'}\\
			&=\esubstdots{(\BText_\Gg(F))}{\BText_\Gg(P_1)}{z_1'}{\BText_\Gg(P_s)}{z_s'}=\BText_\Gg(E').
		\end{align*}
		On the other hand, $\exp(E)=X\,(\Xi(K_1))\,\dots\,(\Xi(K_k))\,(\Xi(L_1))\,\dots\,(\Xi(L_\ell))$, so
		\begin{align*}
			\exp(E)\rew_\Gg {}&R\subst{\Xi(K_1)/y_1,\dots,\Xi(K_k)/y_k,\Xi(L_1)/z_1,\dots,\Xi(L_\ell)/z_\ell}\\
			={}&\Xi(R\subst{K_1/y_1,\dots,K_k/y_k,L_1/z_1,\dots,L_\ell/z_\ell})\\
			={}&\Xi(M'\subst{L_1/z_{s+1}'}\dots\subst{L_\ell/z_{s+\ell}'})=\Xi(\exp(F))=\exp(E');
		\end{align*}
		the first equality holds because the variables $z_1',\dots,z_s'$ do not appear in $R$,
		and the second equality holds because the variables $z_{s+1}',\dots,z_{s+\ell}'$ do not appear in $L_1,\dots,L_\ell$.
		This implies that $\BT_\Gg(\exp(E))=\BT_\Gg(\exp(E'))$, as well as $\exp(E')\rew_\Gg^{r-1}\symb{a,N_1,\dots,N_n}$
		(because $\exp(E)\rew_\Gg^r\symb{a,N_1,\dots,N_n}$).
		Finally, we use the induction hypothesis for $E'$ (with the parameter $r$ decreased by one),
		obtaining that $\BT_\Gg(\exp(E'))=\BTsimpl(\BText_\Gg(E'))$;
		together with aforementioned equalities this implies $\BT_\Gg(\exp(E))=\BTsimpl(\BText_\Gg(E))$, as needed.
		
		The second subcase is when $M$ consist of a variable to which some arguments are applied.
		Because $E$ has no free variables, this variable has to be one of $z_1',\dots,z_s'$; say $z_i'$.
		Moreover, $z_i'$ is of type $\otyp$, so actually there are no arguments, and we simply have $M=z_i'$.
		Let $E'=\esubstdots{P_i}{P_{i+1}}{z_{i+1}'}{P_s}{z_s'}$.
		Observe that
		\begin{align*}
			\BText_\Gg(E)={}&\esubstdots{z_i'}{\BText_\Gg(P_1)}{z_1'}{\BText_\Gg(P_s)}{z_s'}\\
			\simpl{}&\esubstdots{\BText_\Gg(P_i)}{\BText_\Gg(P_{i+1})}{z_{i+1}'}{\BText_\Gg(P_s)}{z_s'}=\BText_\Gg(E'),
		\end{align*}
		so $\BTsimpl(\BText_\Gg(E))=\BTsimpl(\BText_\Gg(E'))$.
		Simultaneously
		\begin{align*}
			\exp(E)=\Xi(z_i)=P_i\subst{P_{i+1}/z_{i+1}'}\dots\subst{P_s/z_s'}=\exp(E').
		\end{align*}
		We use the induction hypothesis for $E'$ (with the parameter $r$ unchanged, and with the parameter $s$ decreased by $i\geq 1$) obtaining that
		$\BT_\Gg(\exp(E'))=\BTsimpl(\BText_\Gg(E'))$;
		in consequence $\BT_\Gg(\exp(E))=\BTsimpl(\BText_\Gg(E))$, as needed.
		
		Finally, the third subcase is when $M$ starts with a node constructor (i.e., is of the form $\symb{b,K_1,\dots,K_k}$).
		Because $\Xi(M)=\exp(E)\rew_\Gg^r\symb{a,N_1,\dots,N_n}$, necessarily $r=0$ and $M=\symb{a,K_1,\dots,K_n}$ with $N_i=\Xi(K_i)$ for all $i\in\scope{n}$.
		Let $E_i=\Upsilon(K_i)$ for all $i\in\scope{n}$; then $N_i=\exp(E_i)$.
		We have $\BT_\Gg(\exp(E))=\symb{a,\BT_\Gg(\exp(E_1)),\dots,\BT_\Gg(\exp(E_n))}$ and
		\begin{align*}
			\BText_\Gg(E)={}&\esubstdots{\symb{a,\BText_\Gg(K_1),\dots,\BText_\Gg(K_n)}}{\BText_\Gg(P_1)}{z_1'}{\BText_\Gg(P_s)}{z_s'}\\
			\simpl{}&\symb{a,\esubstdots{(\BText_\Gg(K_1))}{\BText_\Gg(P_1)}{z_1'}{\BText_\Gg(P_s)}{z_s'},\dots,\\
				&\hspace{8em}\esubstdots{(\BText_\Gg(K_n))}{\BText_\Gg(P_1)}{z_1'}{\BText_\Gg(P_s)}{z_s'}}\\
			={}&\symb{a,\BText_\Gg(E_1),\dots,\BText_\Gg(E_n)},
		\end{align*}
		which implies that
		\begin{align*}
			\BTsimpl(\BText_\Gg(E))&=\BTsimpl(\symb{a,\BText_\Gg(E_1),\dots,\BText_\Gg(E_n)})\\
			&=\symb{a,\BTsimpl(\BText_\Gg(E_1)),\dots,\BTsimpl(\BText_\Gg(E_n))}.
		\end{align*}
		We use the hypothesis of coinduction; it says that $\BT_\Gg(\exp(E_i))=\BTsimpl(\BText_\Gg(E_i))$ for all $i\in\scope{n}$.
		By the above it follows that $\BT_\Gg(\exp(E))=\BTsimpl(\BText_\Gg(E))$, as needed.
		
		This finishes the case when the root of $\BT_\Gg(\exp(E))$ is labeled by some $a\in\Sigma$.
		It remains to consider the case when $\BT_\Gg(\exp(E))=\symb{\omega}$.
		We remark that this case is actually useless, because \cref{std2ext} is used in this paper only for parity recursion schemes, and they do not create $\omega$-labeled nodes.
		Nevertheless, let us sketch a proof also for this case, necessary for the current statement of the \lcnamecref{std2ext-aux}.
		The reason for $\BT_\Gg(\exp(E))=\symb{\omega}$ is that there exists an infinite sequence of terms $K_0,K_1,K_2,\dots$ with $K_0=\exp(E)$
		such that $K_i\rew_\Gg K_{i+1}$ for all $i\in\Nat$.
		The same argumentation as above, which was originally applied to a finite sequence of $\rew_\Gg$ reductions from $\exp(E)$ to $\symb{a,N_1,\dots,N_n}$,
		can be equally well applied to the infinite sequence of $\rew_\Gg$ reductions starting in $\exp(E)$.
		As a result, we obtain an infinite sequence of extended terms $E_0,E_1,E_2,\dots$ with $E_0=E$.
		When the first subcase was used for some $E_i$ (i.e., when the head of $E_i$ is a nonterminal), we have $\BText_\Gg(E_i)=\BText_\Gg(E_{i+1})$,
		and when the second subcase was used for some $E_i$ (i.e., when the head of $E_i$ is a variable), we have $\BText_\Gg(E_i)\simpl\BText_\Gg(E_{i+1})$;
		notice that the third subcase, concerning a node constructor, is impossible for an infinite sequence.
		If the second subcase occurs infinitely often, we have infinitely many $\simpl$ reductions starting in $\BText_\Gg(E)$, which implies that $\BTsimpl(\BText_\Gg(E))=\symb{\omega}$;
		we are done in this case.
		Thus, suppose otherwise: there is an index $j\in\Nat$ such that for all $i\geq j$ the first subcase is applied.
		For all $i\geq j$ let us write $E_i=\Upsilon_i(M_i)$, where $M_i$ is a non-extended term, and $\Upsilon_i$ appends some explicit substitutions.
		Looking more precisely at the first subcase above, we can observe that $\Upsilon_{i+1}(M_{i+1})=\Upsilon_i(\Upsilon_i'(M_{i+1}))$
		where again $\Upsilon'_i$ appends some explicit substitutions, and that $M_i\erew_\Gg\Upsilon'_i(M_{i+1})$.
		If infinitely many among $\Upsilon'_i$ are nonempty (i.e., each of them appends at least one explicit substitution),
		the extended tree $\BText_\Gg(E_j)$ starts with infinitely many explicit substitutions,
		$\BText_\Gg(E_j)=\esubst{\esubst{\esubst{\dots}{L_3}{z_3}}{L_2}{z_2}}{L_1}{z_1}$.
		No $\simpl$ reduction starts in such an extended tree, so $\BTsimpl(\BText_\Gg(E))=\BTsimpl(\BText_\Gg(E_j))=\symb{\omega}$.
		Otherwise, there is an index $k\in\Nat$ such that all $\Upsilon_i'$ for $i\geq k$ are empty, not adding any more explicit substitutions.
		Then $M_i\erew_\Gg M_{i+1}$ for all $i\geq k$, so $\BText_\Gg(M_k)=\symb{\omega}$.
		In consequence, $\BText_\Gg(E_k)=\Upsilon_k(\symb{\omega})\simpl\symb{\omega}$ and $\BTsimpl(\BText_\Gg(E))=\BTsimpl(\BText_\Gg(E_k))=\symb{\omega}$.
	\end{proof}

\section{Proof of Lemma~\ref{ext2trans}}\label{app:ext2trans}

	We start the proof of \cref{ext2trans} by showing two auxiliary \lcnamecrefs{ext2trans-aux}.
	We use an additional notation: for a function $Z$, for a sequence $(z_1',\dots,z_\ell')$, and for a function $A\in D_d^{\scope{\ell}}$,
	we write $Z\mapch{(z_1',\dots,z_\ell')\mapsto A}$ for $Z\mapch{z_i'\mapsto A(\ell+1-i)\mid i\in\scope{\ell}}$.
	
	\begin{lemma}\label{ext2trans-aux}
		Let $\Gg=(\Xx,X_0,\Sigma_d,\Rr)$ be a parity recursion scheme, let $\Zz$ be a set of variables of type $\otyp$, let $z_1',\dots,z_\ell'\not\in\Zz$ be variables of type $\otyp$,
		let $Z\colon\Zz\to D_d$,
		let $N$ be a term over $(\Xx,\Zz,\Sigma_d)$ of type $\otyp^\ell\arr\otyp$,
		let $L_1,\dots,L_\ell$ be terms over $(\Xx,\Zz,\Sigma_d)$ of type $\otyp$,
		let $T$ be an extended tree over $(\Zz\cup\set{z_1',\dots,z_\ell'},\Sigma_d)$, and let $U_1,\dots,U_\ell$ be extended trees over $(\Zz,\Sigma_d)$.
		If
		\begin{align}\label{ass-L}
			\trt_d(Z\restr_r,U_i)=\BT_{\Gg^\dag}(\tr_d(\emptyset,Z\restr_r,L_i))
		\end{align}
		for all $i\in\scope{\ell}$ and $r\in D_d$, and
		\begin{align}\label{ass-N}
			\trt_d(Z\mapch{(z_1',\dots,z_\ell')\mapsto A},T)=\BT_{\Gg^\dag}(\tr_d(A,Z,N))
		\end{align}
		for all $A\in D_d^{\scope{\ell}}$, then
		\begin{align}\label{conc-NL}
			\trt_d(Z,\esubstdots{T}{U_1}{z_1'}{U_\ell}{z_\ell'})=\BT_{\Gg^\dag}(\tr_d(\emptyset,Z,N\,L_1\,\dots\,L_\ell)).
		\end{align}
	\end{lemma}
	
	\begin{proof}
		Induction on $\ell$.
		For $\ell=0$ the thesis, \cref{conc-NL}, says the same as one of the assumptions, \cref{ass-N} for $A=\emptyset$ (i.e., for the only function $A$ in $D_d^\emptyset$).
		Suppose that $\ell\geq 1$.
		In this case, we take $N'=N\,L_1$ and $T'=\esubst{T}{U_1}{z_1'}$,
		and we apply the induction hypothesis for the terms $N',L_2,\dots,L_\ell$ and for the extended trees $T',U_2,\dots,U_\ell$.
		The thesis of the \lcnamecref{ext2trans-aux}, \cref{conc-NL}, says the same as the thesis of the induction hypothesis.
		Also assumptions of the \lcnamecref{ext2trans-aux} described by \cref{ass-L} for $i\in\set{2,\dots,\ell}$ and all $r\in D_d$ can be directly passed to the induction hypothesis.
		We thus only need to prove an analogue of \cref{ass-N} for the induction hypothesis, that is,
		\begin{align}\label{ass-N-IH}
			\trt_d(Z\mapch{(z_2',\dots,z_\ell')\mapsto A'},T')=\BT_{\Gg^\dag}(\tr_d(A',Z,N'))
		\end{align}
		for all $A'\in D_d^{\scope{\ell-1}}$.
		To this end, fix some $A'\in D_d^{\scope{\ell-1}}$.
		Let $Z'=Z\mapch{(z_2',\dots,z_\ell')\mapsto A'}$.
		Recall that
		\begin{align*}
			\trt_d(Z',T')=\trt_d(Z',\esubst{T}{U_1}{z_1'})=\symb{\Eve,1,T_1^{U_1},T_2^{U_1},\dots,T_d^{U_1},T_{2d}},
		\end{align*}
		where $T_r^{U_1}=\symb{\Adam,1,T_r,\symb{\Eve,r,\trt_d(Z'\restr_r,{U_1})}}$ for $r\in\scope{d}$
		and $T_r=\trt_d(Z'\mapch{z_1'\mapsto r},T)$ for $r\in D_d$, and
		\begin{align*}
			\tr_d(A',Z,N')=\tr_d(A',Z,N\,L_1)=\symb{\Eve,1,N_1^{L_1},N_2^{L_1},\dots,N_d^{L_1},N_{2d}},
		\end{align*}
		where $N_r^{L_1}=\symb{\Adam,1,N_r,\symb{\Eve,r,\tr_d(\emptyset,Z\restr_r,L_1)}}$ for $r\in\scope{d}$
		and $N_r=\tr_d(A'\mapch{\ell\mapsto r},\allowbreak Z,N)$ for $r\in D_d$.
		Observe that $Z'\mapch{z_1'\mapsto r}=Z\mapch{(z_1',z_2',\dots,z_\ell')\mapsto A'\mapch{\ell\mapsto r}}$, so by \cref{ass-N}, where we take $A'\mapch{\ell\mapsto r}$ as $A$,
		it follows that $T_r=\BT_{\Gg^\dag}(N_r)$ for all $r\in D_d$.
		Moreover, because variables $z_2',\dots,z_\ell'$ do not appear in $U_1$, it holds that $\trt_d(Z'\restr_r,{U_1})=\trt_d(Z\restr_r,U_1)$;
		using additionally \cref{ass-L} for $i=1$ we obtain that
		\begin{align*}
			T_r^{U_1}=\symb{\Adam,1,\BT_{\Gg^\dag}(N_r),\symb{\Eve,r,\BT_{\Gg^\dag}(\tr_d(\emptyset,Z\restr_r,L_1))}}=\BT_{\Gg^\dag}(N_r^{L_1})
		\end{align*}
		for all $r\in\scope{d}$.
		This gives immediately \cref{ass-N-IH}.
	\end{proof}

	The second auxiliary \lcnamecref{trans-subst-com} says that the $\tr_d$ function commutes with substitution:

	\begin{lemma}\label{trans-subst-com}
		Let $R\subst{K_1/y_1,\dots,K_k/y_k}$ be a term over $(\Xx,\Zz,\Sigma_d)$, let $A\in D_d^{\scope{\gar(R)}}$, and let $Z\in D_d^\Zz$.
		Then
		\begin{align*}
			\tr_d(A,Z,R\subst{K_1/y_1,\dots,K_k/y_k}) &\\
				&\hspace{-4em}= (\tr_d(A,Z,R))\subst{\tr_d(B,Z,K_i)/y^{\dag_d}_{i,B}\mid i\in\scope{k},B\in D_d^{\scope{\gar(K_i)}}}.
		\end{align*}
	\end{lemma}

	\begin{proof}
		A straightforward induction on the structure of $R$.
	\end{proof}

	We now have a counterpart of \cref{ext2trans} that is suitable for a (co)inductive proof;
	the original statement can be obtained by taking the starting nonterminal $X_0$ as $M$.

	\begin{lemma}
		Let $\Gg=(\Xx,X_0,\Sigma_d,\Rr)$ be a parity recursion scheme.
		For every set $\Zz$ of variables of type $\otyp$, for every function $Z\colon\Zz\to D_d$, and for every term $M$ over $(\Xx,\Zz,\Sigma_d)$ of type $\otyp$,
		it holds that $\trt_d(Z,\BText_\Gg(M))=\BT_{\Gg^\dag}(\tr_d(\emptyset,Z,M))$.
	\end{lemma}
	
	\begin{proof}
		The proof is by coinduction on the structure of $\BText_\Gg(M)$ and, internally,
		by induction on the number of $\erew_\Gg$ steps from $M$ needed to generate the next node of $\BText_\Gg(M)$.
		We have three cases, depending on the shape of $M$.
		
		The first case is when $M$ consists of a nonterminal $X$ to which some arguments are applied.
		We can write $M=X\,K_1\,\dots\,K_k\,L_1\,\dots\,L_\ell$, where $\ell=\gar(X)$.
		Let $X\,y_1\,\dots\,y_k\,z_1\,\dots\,z_\ell\to R$ be the rule for $X$.
		Then, denoting $R\subst{K_1/y_1,\dots,K_k/y_k,z_1'/z_1,\dots,z_\ell'/z_\ell}$ as $M'$ for some fresh variables $z_1',\dots,z_\ell'\not\in\Zz$,
		we have that
		\begin{align*}
			M\erew_\Gg \esubstdots{M'}{L_1}{z_1'}{L_\ell}{z_\ell'},
		\end{align*}
		which implies that
		\begin{align*}
			\BText_\Gg(M)&=\esubstdots{(\BText_\Gg(M'))}{\BText_\Gg(L_1)}{z_1'}{\BText_\Gg(L_\ell)}{z_\ell'}.
		\end{align*}
		We thus have to prove that
		\begin{align*}
			\trt_d(Z,\esubstdots{(\BText_\Gg(M'))}{\BText_\Gg(L_1)}{z_1'}{\BText_\Gg(L_\ell)}{z_\ell'})&\\
				&\hspace{-6em}= \BT_{\Gg^\dag}(\tr_d(\emptyset,Z,X\,K_1\,\dots\,K_k\,L_1\,\dots\,L_\ell)).
		\end{align*}
		This is the conclusion of \cref{ext2trans-aux}, when applied to $\BText_\Gg(M')$ as $T$, and $\BText_\Gg(L_i)$ as $U_i$ for $i\in\scope{\ell}$,
		and $X\,K_1\,\dots\,K_k$ as $N$.
		It remains to check assumptions of the \lcnamecref{ext2trans-aux}, namely
		\begin{align}\label{ass-L-use}
			\trt_d(Z\restr_r,\BText_\Gg(L_i))=\BT_{\Gg^\dag}(\tr_d(\emptyset,Z\restr_r,L_i))
		\end{align}
		for all $i\in\scope{\ell}$ and $r\in D_d$, and
		\begin{align}\label{ass-N-use}
			\trt_d(Z\mapch{(z_1',\dots,z_\ell')\mapsto A},\BText_\Gg(M'))=\BT_{\Gg^\dag}(\tr_d(A,Z,X\,K_1\,\dots\,K_k))
		\end{align}
		for all $A\in D_d^{\scope{\ell}}$.
		The first part, \cref{ass-L-use}, follows directly from the hypothesis of coinduction (because $\BText_\Gg(L_i)$ is a proper subtree of $\BText_\Gg(M)$).
		In order to prove \cref{ass-N-use}, fix some $A\in D_d^{\scope{\ell}}$.
		First observe that
		\begin{align}\label{eq:1}
			\trt_d(Z\mapch{(z_1',\dots,z_\ell')\mapsto A},\BText_\Gg(M'))=\BT_{\Gg^\dag}(\tr_d(\emptyset,Z\mapch{(z_1',\dots,z_\ell')\mapsto A},M')).
		\end{align}
		Indeed, if $\ell\geq 1$, this is the hypothesis of coinduction (because $\BText_\Gg(M')$ is a proper subtree of $\BText_\Gg(M)$),
		and if $\ell=0$, this is the hypothesis of the internal induction (because $M\erew_\Gg M'$).
		Next, recall that
		\begin{align*}
			\tr_d(A,Z,X\,K_1\,\dots\,K_k)&= X_A^{\dag_d}\,(\tr_d(B,Z,K_1))_{B\in D_d^{\scope{\gar(K_1)}}}\,\dots\,(\tr_d(B,Z,K_k))_{B\in D_d^{\scope{\gar(K_k)}}}\\
			&\hspace{-7em}\rew_{\Gg^\dag} \tr_d(\emptyset,\mapch{(z_1,\dots,z_\ell)\mapsto A},R)\subst{\tr_d(B,Z,K_i)/y_{i,B}^{\dag_d}\mid i\in\scope{k},B\in D_d^{\scope{\gar(K_i)}}},
		\end{align*}
		meaning that the tree generated by $\Gg^\dag$ from these terms is the same.
		Using this fact and \cref{eq:1}, we see that in order to obtain \cref{ass-N-use} it is enough to prove that
		\begin{align}\label{eq:2}
			\tr_d(\emptyset,Z\mapch{(z_1',\dots,z_\ell')\mapsto A},R\subst{K_1/y_1,\dots,K_k/y_k,z_1'/z_1,\dots,z_\ell'/z_\ell})&\nonumber\\
				&\hspace{-27em}=\tr_d(\emptyset,\mapch{(z_1,\dots,z_\ell)\mapsto A},R)\subst{\tr_d(B,Z,K_i)/y_{i,B}^{\dag_d}\mid i\in\scope{k},B\in D_d^{\scope{\gar(K_i)}}}.
		\end{align}
		Below we assume that the variables $y_1,\dots,y_k$ do not belong to $Z\cup\set{z_1',\dots,z_\ell'}$.
		(If they do, simply change their names to some fresh variables $y_1',\dots,y_k'$,
		and conduct the proof of the above equality for those variables and for $R\subst{y'_1/y_1,\dots,y'_k/y_k}$ in place of $R$.
		Anyway, these are only variables for which we substitute something, so changing their names does not influence the result;
		both sides of the above equality remain unchanged.)
		Observe that
		\begin{align}\label{eq:3}
			R\subst{K_1/y_1,\dots,K_k/y_k,z_1'/z_1,\dots,z_\ell'/z_\ell}=R\subst{z_1'/z_1,\dots,z_\ell'/z_\ell}\subst{K_1/y_1,\dots,K_k/y_k}.
		\end{align}
		Observe also that
		\begin{align}\label{eq:4}
			\tr_d(\emptyset,\mapch{(z_1,\dots,z_\ell)\mapsto A},R)&=\tr_d(\emptyset,\mapch{(z_1',\dots,z_\ell')\mapsto A},R\subst{z_1'/z_1,\dots,z_\ell'/z_\ell})\nonumber\\
			&=\tr_d(\emptyset,Z\mapch{(z_1',\dots,z_\ell')\mapsto A},R\subst{z_1'/z_1,\dots,z_\ell'/z_\ell});
		\end{align}
		above, in the first equality we have just renamed variables $z_1,\dots,z_\ell$ to $z_1',\dots,z_\ell'$;
		in the second equality we use the fact that variables from $\Zz$ (i.e., from the domain of $Z$) do not appear in $R\subst{z_1'/z_1,\dots,z_\ell'/z_\ell}$.
		Likewise, because variables $z_1',\dots,z_\ell'$ do not appear in the terms $K_i$, we have
		\begin{align}\label{eq:5}
			\tr_d(B,Z,K_i)=\tr_d(B,Z\mapch{(z_1',\dots,z_\ell')\mapsto A},K_i)
		\end{align}
		for all $i\in\scope{k}$ and all $B\in D_d^{\scope{\gar(K_i)}}$.
		Applying \cref{eq:3,,eq:4,,eq:5} to \cref{eq:2}, the latter \lcnamecref{eq:2} changes to
		\begin{align*}
			\tr_d(\emptyset,Z\mapch{(z_1',\dots,z_\ell')\mapsto A},R\subst{z_1'/z_1,\dots,z_\ell'/z_\ell}\subst{K_1/y_1,\dots,K_k/y_k})&\\
				&\hspace{-27em}=\tr_d(\emptyset,Z\mapch{(z_1',\dots,z_\ell')\mapsto A},R\subst{z_1'/z_1,\dots,z_\ell'/z_\ell})\\
					&\hspace{-23em}\subst{\tr_d(B,Z\mapch{(z_1',\dots,z_\ell')\mapsto A},K_i)/y_{i,B}^{\dag_d}\mid i\in\scope{k},B\in D_d^{\scope{\gar(K_i)}}}.
		\end{align*}
		This equality follows directly from \cref{trans-subst-com}, which finishes the proof in the first case.

		The second case is when $M$ consist of a variable to which some arguments are applied.
		All free variables of $M$ are elements of $\Zz$, so they are of type $\otyp$, hence necessarily $M=z$ for some $z\in\Zz$.
		If $Z(z)$ is odd, we have
		\begin{align*}
			\trt_d(Z,\BText_\Gg(M))=\trt_d(Z,z)=\top=\BT_{\Gg^\dag}(\downVdash)=\BT_{\Gg^\dag}(\tr_d(\emptyset,Z,M));
		\end{align*}
		likewise if $Z(z)$ is even, but with $\bot$ and $\upVdash$ instead of $\top$ and $\downVdash$.
		
		Finally, the third case is when $M$ starts with a node constructor, that is, it is of the form $\symb{\player,p,K_1,\dots,K_k}$.
		Then $\trt_d(Z,\BText_\Gg(K_i))=\BT_{\Gg^\dag}(\tr_d(\emptyset,Z,K_i))$ by the hypothesis of coinduction, so
		\begin{align*}
			\trt_d(Z,\BText_\Gg(M))&=\trt_d(Z,\symb{\player,p,\BText_\Gg(K_1),\dots,\BText_\Gg(K_k)})\\
				&=\symb{\player,p,\trt_d(Z,\BText_\Gg(K_1)),\dots,\trt_d(Z,\BText_\Gg(K_k))}\\
				&=\symb{\player,p,\BT_{\Gg^\dag}(\tr_d(\emptyset,Z,K_1)),\dots,\BT_{\Gg^\dag}(\tr_d(\emptyset,Z,K_k))}\\
				&=\BT_{\Gg^\dag}(\symb{\player,p,\tr_d(\emptyset,Z,K_1),\dots,\tr_d(\emptyset,Z,K_k)})\\
				&=\BT_{\Gg^\dag}(\tr_d(\emptyset,Z,\symb{\player,p,K_1,\dots,K_k}))=\BT_{\Gg^\dag}(\tr_d(\emptyset,Z,M)).\qedhere
		\end{align*}
	\end{proof}

\end{document}